%% file: main-counting.tex
\documentclass[11pt]{article}
\usepackage{times} 
\usepackage{helvet}
\usepackage{courier}
\usepackage[hyphens]{url} 
\usepackage{graphicx}
\usepackage{graphicx}
\usepackage[numbers]{natbib}
\usepackage{caption}
\usepackage{thmtools, thm-restate}
\usepackage{arydshln}
\usepackage[hidelinks,pagebackref]{hyperref}
\hypersetup{pagebackref=true}
\renewcommand*{\backref}[1]{}
\renewcommand*{\backrefalt}[4]{%
\ifcase #1%
\or
$\rightarrow$~Page~#2.%
\else
$\rightarrow$~Pages~#2.%
\fi
}

\title{On the Robustness of Winners:\\ Counting Briberies in Elections}

\author{Niclas Boehmer\\
  TU Berlin, Germany\\ 
  \small{\texttt{niclas.boehmer@tu-berlin.de}}
\and
  Robert Bredereck\\
  Humboldt-Universit\"at zu Berlin, Germany\\
  \small{\texttt{robert.bredereck@hu-berlin.de}}
\and
  Piotr Faliszewski\\
  AGH University, Poland\\
  \small{\texttt{faliszew@agh.edu.pl}}
\and
  Rolf Niedermeier\\
  TU Berlin, Germany\\
  \small{\texttt{rolf.niedermeier@tu-berlin.de}}
}

\usepackage{multirow}

\usepackage{url}
\usepackage{graphicx}
\usepackage{booktabs}

\usepackage{amsmath}
\usepackage{amsthm}
\usepackage{amssymb}
\usepackage{subcaption}
\usepackage{bm}

\usepackage{dsfont}
\usepackage{mathtools}

\usepackage{comment}
\usepackage{pgfplots}
\pgfplotsset{compat=newest}
\usepackage{cleveref}
\usepackage{subcaption}
\usetikzlibrary{matrix}
\usepackage[textsize=scriptsize]{todonotes}
\newcommand{\todoi}[1]{}

\newcommand{\increase}{{{\mathrm{inc}}}}

\usepackage{nicefrac}

\newtheorem{theorem}{Theorem}

\newtheorem{corollary}{Corollary}

\newtheorem{lemma}[theorem]{Lemma}

\newtheorem{conclusion}{Finding}

\newcommand{\removememaybe}[1]{}

\newcommand{\calR}{\mathcal{R}}

\newcommand{\np}{{\mathrm{NP}}}

\newcommand{\fpt}{{\mathrm{FPT}}}

\newcommand{\xp}{{\mathrm{XP}}}
\newcommand{\wone}{{\mathsf{W[1]}}}

\newcommand{\p}{{\mathrm{P}}}

\newcommand{\swap}{{\mathrm{sw}}}

\DeclareMathOperator{\score}{score}

\newcommand{\pref}{\ensuremath{\succ}}

\newcommand{\Winner}{\textsc{Winner}}
\newcommand{\bordaOWAwinner}[1][]{%
	$\beta$%
	\ifx&#1&%
	\else
	(#1)%
	\fi%
	-\Winner%
}%

\allowdisplaybreaks

\clearpage{}%

\providecommand{\np}{\ensuremath{\mathrm{NP}}}

\providecommand{\p}{\ensuremath{\mathrm{P}}}

\providecommand{\fpt}{\ensuremath{\mathrm{FPT}}}                                
     
\providecommand{\xp}{\ensuremath{\mathrm{XP}}}

\providecommand{\wone}{\ensuremath{\mathrm{\w[1]}}}

\providecommand{\sharpp}{\ensuremath{\#\p}}
\providecommand{\sharpwone}{\ensuremath{\#\wone}}

\providecommand{\calR}{{\mathcal{R}}}

\providecommand{\score}{{\mathrm{score}}}

\newcommand{\scoreof}[1]{\ensuremath{s(#1)}}
\newcommand{\finScP}{\ensuremath{s^*}}
\newcommand{\valSh}[2]{\ensuremath{\nu(#1,#2)}}
\newcommand{\topSh}[2]{\ensuremath{\tau(#1,#2)}}
\newcommand{\vgcSwapPlu}{\ensuremath{\text{vgc}^{\text{Swap}}_{\text{Plurality}}}}
\newcommand{\vgcShifPlu}{\ensuremath{\text{vgc}^{\text{Shift}+}_{\text{Plurality}}}}
\newcommand{\vgcDShifPlu}{\ensuremath{\text{vgc}^{\text{Shift}-}_{\text{Plurality}}}}

\newcommand{\first}{}    
\def\first/{red}
\newcommand{\second}{}    
\def\second/{blue}
\newcommand{\third}{}    
\def\third/{black}
\newcommand{\fourth}{}    
\def\fourth/{green}
\clearpage{}%

\usepackage[mode=buildnew]{standalone}
\begin{document}
	
	\maketitle
	
	\begin{abstract}
		We study the parameterized complexity of counting variants of
		\textsc{Swap-} and \textsc{Shift-Bribery} problems, focusing on the
		parameterizations by the number of swaps and the number of
		voters. We show experimentally that \textsc{Swap-Bribery}
		offers a new approach to the robustness analysis of elections.
	\end{abstract}
	
	\section{Introduction}\label{sec:intro}
	
	Consider a university department
	which is about to hire a new professor. There are~$m$
	candidates 
	and the head of the department 
	decided
	to choose the 
	winner by Borda voting.
	Each faculty member (i.e., each voter)
	ranked the candidates from the most to the least appealing one, each
	candidate received $m-i$ points for each vote where he or she was
	ranked as the $i$-th best, and 
	the candidate with the highest score was selected.
	However, after the results were announced, some voters started
	wondering if, perhaps, some \emph{other} voters accidently ``misranked''
	some of the candidates
	(worrying about 
	mistakes in the votes is an old
	democratic tradition).
	For instance, 
	if some voter viewed two candidates as very similar, then he or she
	could have ranked them either way, depending on an
	impulse. 
	Or, some voter would have ranked two candidates differently if he or
	she had more information on their merits
	(this is particularly likely for low-ranked
	candidates\footnote{If some voters rank these candidates highly, then
		even a single point may mean the difference between winning the
		election or not.}).
	It is, thus, natural to ask for the
	probability of changing 
	the election outcome by making
	some random swaps.
	Indeed, this approach was recently pursued by
	\citet{bau-hog:c:counting-swap-bribery-evaluation} and we follow-up on
	it, but with a somewhat different focus (we will discuss this difference
	together with other related work).
	
	Specifically, 
	for each $r \in \mathbb{N}$ and each candidate~$c$, we let~$P_c(r)$ be
	the probability that~$c$ wins an election
	obtained 
	by making~$r$ random swaps of candidates ranked on adjacent
	positions in the votes (we refer to such elections as being at swap
	distance $r$ from the original one). Such values can be quite
	useful. 
	For example, if for each~$r$ we had (some estimate of) the
	probability that 
	in total there are~$r$ accidental swaps in the
	votes,
	then we could compute the probability of each
	candidate's victory. If it were small for the original
	winner, then we might want to recount the votes or reexamine the
	election process.  The values~$P_c(r)$ are also useful without the
	distribution of~$r$'s.  For example, we may want to find the smallest
	number of swaps for which the probability of the original winner's
	victory drops below some value (such as 50\%) or for which he or she
	is no longer the most probable winner.  As we show in our experiments,
	this approach 
	provides new insights on the robustness of election results.
	
	To determine the value~$P_c(r)$, we need to divide the number of
	elections at swap distance $r$ where $c$ wins, by the total number of
	elections at this distance. While computing the latter is easy---at
	least in the sense that there is a polynomial-time algorithm for this
	task---computing the former requires solving the counting variant of
	the \textsc{Swap-Bribery} problem (denoted \textsc{\#Swap-Bribery}).
	Briefly put, in the decision variant of the
	problem, 
	we ask if
	it is possible to ensure that a designated candidate
	wins a given election by making 
	$r$ swaps of adjacent candidates in the votes (we assume the unit
	prices setting; see Section~\ref{sec:prelim}). In the counting
	variant, we ask how many ways there are to achieve this effect (using
	exactly~$r$ swaps).  Unfortunately, already the decision variant
	is $\np$-hard for many 
	voting rules,
	and we show that 
	the counting one 
	is 
	hard even
	for Plurality.  On the positive side, we can get a good estimate of
	$P_c(r)$ by 
	sampling.
	
	We also consider the \textsc{Shift-Bribery} problem,
	a variant of \textsc{Swap-Bribery}
	where 
	we can only shift the designated candidate forward (in the
	constructive case) or backward (in the destructive one, where the goal
	is to ensure that the designated candidate loses).  These
	problems 
	also can be used to evaluate robust\-ness of election results but, to
	maintain focus, in our experiments we only consider
	\textsc{Swap-Bribery}.  Yet, we include \textsc{Shift-Bribery} in
	our 
	complexity analysis because it 
	illustrates some
	interesting phenomena.
	\bigskip
	
        \subsection{Main Contributions}
	We focus on \textsc{\#Swap-} and \textsc{\#Shift-Bribery} for the
	Plurality and Borda voting rules (for unit prices). 
	We consider their computational complexity for parameterizations by
	the number of unit swaps/shifts (which we refer to as the swap/shift
	radius) and by the number of voters (see Table~\ref{tab:results}).
	We also present experiments, where we use \textsc{\#Swap-Bribery} to
	evaluate the robustness of election results.  Our main results are as
	follows:
	\begin{enumerate}
		\item For Plurality, \textsc{Swap-Bribery} is known to be in $\p$, but
		we show that the counting variant is $\sharpp$-hard, and even
		$\sharpwone$-hard for the parameterization by the swap radius.
		
		\item For Borda, hardness results for \textsc{\#Swap-Bribery} follow 
		from those
		for \textsc{\#Shift-Bribery}, which themselves are intriguing: 
		E.g.,
		the destructive variant parameterized by the shift
		radius
		is
		$\sharpwone$-hard, but the constructive one is in $\fpt$; yet,
		in the decision setting the former is easier. 
		
		\hyphenation{Szu-fa}
		
		\item Using 
		sampling, 
		we estimate the candidate's winning probabilities in elections from
		a dataset generated by
		\citet{szu-fal-sko-sli-tal:c:map-of-elections}.  One of the
		high-level conclusions 
		is that the score differences between the election winners and the
		runners-up can be quite disconnected from their strengths (measured
		using \textsc{\#Swap-Bribery}).
	\end{enumerate}
	Some  proofs and analyses are available in the appendix.

	\begin{table}
		\centering
		\setlength{\tabcolsep}{5pt}
		\begin{tabular}{r|cc|cc}
			\toprule
			& \multicolumn{2}{c|}{Plurality}  & \multicolumn{2}{|c}{Borda} \\
			& decision & counting & decision & counting \\
			\midrule
			\textsc{} &      & $\sharpp$-hard              & 
			$\np$-hard      & $\sharpp$-hard      \\
			\textsc{Swap-Bribery}      & $\p$ & $\sharpwone$-hard$(r)$ & 
			$\fpt(r)$       & ?    \\
			&      & $\fpt(n)$            & $\wone$-hard$(n)$ & 
			$\sharpwone$-hard$(n)$\\[1pt]
			\midrule
			\textsc{Constructive}        &       &          & $\np$-hard         & 
			$\sharpp$-hard \\
			\textsc{Shift-Bribery}        & $\p$  &  $\p$    & $\fpt(r)$          & 
			$\fpt(r)$  \\
			\textsc{}       &       &          & $\wone$-hard$(n)$  & 
			$\sharpwone$-hard$(n)$ \\[1pt]
			\cdashline{1-5}
			\textsc{Destructive}         & \rule{0cm}{0.3cm}      &          &       
			& $\sharpp$-hard \\
			\textsc{Shift-Bribery}        & $\p$  &  $\p$    &  $\p$ & 
			$\sharpwone$-hard$(r)$ \\
			\textsc{}       &       &          &       & 
			$\sharpwone$-hard$(n)$ \\
			\bottomrule
		\end{tabular}
		\caption{\label{tab:results}(Parameterized) complexity
                  of \textsc{Swap-} and \textsc{Shift-Bribery} with
                  unit prices; $r$~and~$n$ refer to the
                  parameterizations by the swap/shift radius and by
                  the number of voters, respectively.  Results for the
                  counting variants are due to this paper (however,
                  see also the work of
                  \citet{bau-hog:c:counting-swap-bribery-evaluation}
                  for results related to $\sharpp$-hardness of
                  \textsc{Swap-Bribery}); the other ones are due to
                  \citet{elk-fal-sli:c:swap-bribery},
                  \citet{bre-che-fal-nic-nie:j:prices-matter,bre-fal-nie-tal:c:shift-bribery-committees},
                  and \citet{fal-kac:j:destructive-shift-bribery}.}
	\end{table}
        
        \subsection{Related Work}
	Our work is most closely related to the papers of
	\citet{haz-aum-kra-woo:j:uncertain-election-outcomes},
	\citet{bac-bet-fal:c:counting-pos-win}, and
	\citet{bau-hog:c:counting-swap-bribery-evaluation}.  Like us, their
	authors study the complexity of computing the probability that a given
	candidate wins, provided that the votes may change according to some
	probability distribution.
	In particular, \citet{haz-aum-kra-woo:j:uncertain-election-outcomes}
	assume that each voter is endowed with an explicitly encoded list of
	possible votes, each with its probability of being cast,
	\citet{bac-bet-fal:c:counting-pos-win} consider elections where the
	votes are partial and all completions are equally likely, and
	\citet{bau-hog:c:counting-swap-bribery-evaluation} consider both these
	models, as well as a third one, where the votes may change according
	to the Mallows noise model~\citep{mal:j:mallows}.
	
	Under the Mallows model, we are given an election---to which
        we refer as the original one---and a parameter $\phi$. Each
        possible election is associated with weight $\phi^d$, where
        $d$ is its swap distance to the original one, and the
        probability of drawing a particular election is proportional
        to its weight.  Thus, the Mallows model is very closely
        related to our approach of counting solutions for \textsc{Swap
          Bribery}. Indeed, the only difference is that we take the
        number of swaps as part of the input (so, intuitively, we view
        each election at this swap distance as equally likely), and in
        the Mallows model
        \citet{bau-hog:c:counting-swap-bribery-evaluation} take $\phi$
        as part of the input and consider all possible swap distances
        (but the probability of drawing an election at a given
        distance is weighted according to the Mallows model with
        $\phi$ as the parameter).
	
	There are two methodological differences between our work and the
	three above-discussed papers.  Foremost, we provide a detailed
	experimental analysis showing that counting variants of
	\textsc{Swap-Bribery} are indeed helpful for evaluating robustness of
	election winners.
	In contrast, \citet{bac-bet-fal:c:counting-pos-win} and
	\citet{bau-hog:c:counting-swap-bribery-evaluation} focus entirely on
	the complexity analysis, whereas
	\citet{haz-aum-kra-woo:j:uncertain-election-outcomes} also provide
	experiments, but their focus is on the running time and memory
	consumption of their algorithm.
	
	The second difference regards the use of parameterized complexity theory.
	Indeed, we believe that we are the first to use a
	parameterized counting complexity analysis---with an explicit focus on
	establishing $\fpt$ and $\sharpwone$-hardness results---in the context
	of elections. However, we do mention that
	\citet{haz-aum-kra-woo:j:uncertain-election-outcomes} and
	\citet{bau-hog:c:counting-swap-bribery-evaluation} consider settings
	where either the numbers of candidates or the numbers of voters are
	fixed constants, so, effectively, they provide $\xp$ algorithms.
	\medskip
	
	\textsc{Swap-} and \textsc{Shift-Bribery} were introduced by
	\citet{elk-fal-sli:c:swap-bribery}. Various authors studied these
	problems for different voting rules (see, e.g., the works of
	\citet{mau-nev-rot-sel:c:shift-bribery-stv} and
	\citet{zhou-guo:c:shift-bribery-fpt-stv} regarding iterative
	elections), sought approximation algorithms
	\citep{elk-fal:c:shift-bribery,fal-man-sor:c:shift-bribery-ptas},
	established parameterized complexity results
	\citep{dor-sch:j:parameterized-swap-bribery,bre-che-fal-nic-nie:j:prices-matter,kno-kou-mni:j:fpt-swap-bribery},
	considered restricted preference
	domains~\citep{elk-fal-gup-roy:c:swap-shift-bribery-single-peaked}, and
	extended the problem in various ways
	\citep{bre-fal-nie-tal:c:shift-bribery-committees,fal-kac:j:destructive-shift-bribery,bau-hog-rey:c:distance-bribery,yan-raj-guo:j:distance-bribery}.
	The idea of using \textsc{Swap-Bribery} to measure the robustness of
	election results is due to \citet{shi-yu-elk:c:robustness}, but is
	also closely related to computing the margin of victory
	\citep{mag-riv-she-wag:c:stv-bribery,car:c:stv-bribery,xia:margin-of-victory,bri-sch-suk:c:mov-tournaments};
	recently it was also applied to 
	committee elections~\citep{bre-fal-kac-nie-sko-tal:c:robustness}.
	
	So far, the complexity of counting problems 
	received fairly limited attention in the context of elections.  In
	addition to the works of
	\citet{haz-aum-kra-woo:j:uncertain-election-outcomes},
	\citet{bac-bet-fal:c:counting-pos-win} and
	\citet{bau-hog:c:counting-swap-bribery-evaluation}, we mention two
	more: \citet{fal-woj:c:counting-control} studied the complexity of
	counting solutions for control problems, whereas
	\citet{ken-kim:c:approx-counting-possible-winner} followed up on the
	work of \citet{bac-bet-fal:c:counting-pos-win} and provided
	approximation algorithms for their setting.
	
	\section{Preliminaries}\label{sec:prelim}
	
	For each integer $k$, by $[k]$ we mean the set $\{1, \ldots, k\}$.
	\medskip
	
	\noindent\textbf{Elections.}\quad
	An election $E = (C,V)$ consists of a set $C = \{c_1, \ldots, c_m\}$ 
	of candidates
	and a collection $V = (v_1, \ldots, v_n)$ of voters.
	Each voter $v_i$ has a preference order,
	which ranks all the candidates from the most to the least desired one
	(we sometimes refer to preference orders as votes).  For a voter $v_i$,
	we write $v_i \colon c_1 \pref c_2 \pref \cdots \pref c_m$ to indicate that
	he or she ranks $c_1$ first, then $c_2$, and so on.
	If we put a subset of candidates in such a description of a preference
	order, then we mean listing its members 
	in an arbitrary order.\medskip
	
	\noindent\textbf{Voting Rules.}\quad
	A voting rule $\calR$ is a function that, given an election,
	returns a set
	of candidates that tie as
	winners. We focus on Plurality and Borda, which assign scores to the
	candidates and select those with the highest ones. Under Plurality,
	each voter gives one point to the top-ranked candidate.  Under Borda,
	each voter gives $|C|-1$ points to the top-ranked candidate, $|C|-2$
	points to the next one, and so on.  We write $\score_E(c)$ to denote
	the score of candidate $c$ in election $E$ (the voting rule will be
	clear from the context).\medskip
	
	\noindent\textbf{Swap Distance.}\quad Let~$u$ and~$v$ be two votes
	over the same candidate set. The swap distance between~$u$ and~$v$,
	denoted~$d_\swap(u,v)$, is the length of the shortest sequence of
	swaps of adjacent candidates whose application transforms~$u$
	into~$v$.  Given elections~$E = (C,V)$ and $E' = (C,V')$, where
	$V = (v_1, \ldots, v_n)$ and $V' = (v'_1, \ldots, v'_n)$, their swap
	distance is 
	$\sum_{i=1}^n d_\swap(v_i,v'_i)$.
	By $R(E,r)$, we denote the set of elections that are at swap
	distance~$r$ from~$E$.
	\medskip
	
	\noindent\textbf{Swap- and Shift-Bribery.}\quad
	Let $\calR$ be a voting rule.  In the decision variant of the
	$\calR$~\textsc{Swap-Bribery} problem, we are given an election
	$E$,
	a designated candidate $p$, and a budget~$r$. Further, for each voter
	and each two candidates~$c$ and~$d$, we have a nonnegative price
	$\pi_v(c,d)$ for swapping them in $v$'s preference order (a swap is
	legal if at the time of its application~$c$ and~$d$ are adjacent).
	We ask if there is an election $E'$ where $p$ is an $\calR$-winner,
	such that $E'$ 
	can be obtained from $E$ by performing a sequence of legal swaps of
	cost at most~$r$.  In the counting variant, we ask for the number of
	such elections, and we require the cost of swaps to be exactly~$r$
	(the last condition is for our convenience and all our results would
	still hold if we asked for cost at most~$r$; the same would be true if
	instead of counting elections where $p$~won, we would count those where
	he or she lost).
	Since we are interested in computing the candidates' probabilities of 
	victory
	in elections at a given swap distance,
	we
	focus on the case where each swap has the same, unit price; thus, we
	usually refer to~$r$ as the swap radius and not as the budget.
	
	\textsc{Constructive Shift-Bribery} is a variant of
	\textsc{Swap-Bribery} where all swaps must involve the designated
	candidate, shifting him or her forward.  \textsc{Destructive
		Shift-Bribery} is defined analogously, except that our goal is to
	preclude the designated candidate's victory, and we can only shift him
	or her backward~\citep{fal-kac:j:destructive-shift-bribery}. %
	Counting variants 
	are defined in a natural way.
	We focus on the case
	where each unit shift
	has a unit price
	and we speak of shift radius~$r$ instead of budget or swap radius.
	
	\medskip
	
	\noindent\textbf{Counting Complexity.}\quad
	We assume basic familiarity with (parameterized) complexity theory,
	including
	classes $\p$, $\np$, $\fpt$, and $\wone$, and reducibility notions.
	
	Let \textsc{X} be a decision problem from $\np$, where for each
	instance we ask if there exists some mathematical object with a given
	property.  In its counting variant,
	traditionally denoted \textsc{\#X}, we ask for the number of such
	objects. For example, in \textsc{Matching} we are given an integer $k$
	and a bipartite graph $G$---with vertex set $U(G) \uplus V(G)$ and
	edge set $E(G)$---and we ask if $G$ contains a matching of size $k$
	(i.e., a set of $k$~edges,
	where no two edges touch the same vertex).
	In
	\textsc{\#Matching} we ask how many such matchings exist.
	
	The class $\sharpp$ is the counting analog of $\np$; a problem belongs
	to $\sharpp$ if it can be expressed as the task of counting accepting
	computations of a nondeterministic polynomial-time Turing machine.  We
	say that a counting problem \textsc{\#A} (polynomial-time) Turing
	reduces to 
	\textsc{\#B} if there exists an algorithm that solves \textsc{\#A} in
	polynomial time, provided
	that it has oracle access to 
	\textsc{\#B}.
	A problem is $\sharpp$-hard if every problem from $\sharpp$ Turing
	reduces to it.
	While \textsc{Matching} is in $\p$, it is well known that
	\textsc{\#Matching} is $\sharpp$-hard, and even
	$\sharpp$-complete~\citep{val:j:permanent}.
	
	$\sharpwone$ relates to $\wone$ in the same way as $\sharpp$ relates
	to $\np$. As examples of $\sharpwone$-hard problems, we mention
	counting size-$k$ cliques in a graph, parameterized
	by~$k$~\citep{flu-gro:j:parameterized-counting-complexity} and
	\textsc{\#Matching}, parameterized by the size of the
	matching~\citep{cur-mar:c:counting-bipartite-matching}.  Formally, 
	$\sharpwone$-hardness is defined using a slightly more general notion
	of a reduction, but for our purposes polynomial-time Turing reductions
	(where the parameters in the queried instances are bounded by a function
	of the parameter in the input instance)
	will suffice.
	
	\section{Algorithms and Complexity Results}\label{sec:complexity}
	
	In this section, we present our results regarding the complexity of
	\textsc{\#Swap-} and \textsc{\#Shift-Bribery}.  We first consider
	Plurality, mostly focusing on the former problem, and then discuss
	Borda, mostly focusing on the latter. 
	
	\subsection{Plurality and \textsc{\#Swap-Bribery}}
	
	We start with bad news. While there is a polynomial-time algorithm for
	the decision variant of \textsc{Plurality
		Swap-Bribery}~\citep{elk-fal-sli:c:swap-bribery}, the counting
	variant is intractable, even with unit prices (for the
	$\sharpp$-hardness, a related result is reported by
	\citet{bau-hog:c:counting-swap-bribery-evaluation}).
	
	\begin{theorem}\label{thm:plurality-swap-radius}
		\textsc{Plurality \#Swap-Bribery} is $\sharpp$-hard and
		$\sharpwone$-hard for the parameterization by the swap radius, even
		for
		unit prices.
	\end{theorem}
	\begin{proof}
		
		We give a reduction from \textsc{\#Matching}.  We will use a swap
		radius bounded by a function of the desired matching size, so we
		will obtain both $\sharpp$- and $\sharpwone$-hardness.
		
		Let $(G,k)$ be an instance of \textsc{\#Matching}, where $G$ is a
		bipartite graph with vertex set $U(G) \uplus V(G)$ and $k$ is the
		size of the matchings that we are to count. Assume that
		$U(G) = \{u_1, \ldots, u_n\}$, $V(G) = \{v_1, \ldots, v_n\}$, and
		$k \leq n$.  To form an election, 
		we let the candidate set be
		$C := U(G) \uplus V(G) \uplus \{p,a,b\} \uplus X$, where
		$X := \{x_1, \ldots, x_{3k+1}\}$. The candidates in
		$U(G) \uplus V(G)$ will model the graph, $p$ will be our designated
		candidate, $a$ and $b$ will control the size of the matching, and
		the candidates in $X$ will block undesirable swaps. We will have
		the following scores of the candidates:
		\begin{align*}
		\forall c\in C\setminus \{ a,b\}\colon& \score_E(c)=n,\\
		\score_E(a) = n-k,\ & \score_E(b) = n+k.
		\end{align*}
		We form the following four groups of voters:
		\begin{enumerate}
			\item For each edge $\{u_i,v_j\} \in E(G)$, there is an \emph{edge} 
			voter~$e_{ij}$
			with preference order
			$e_{ij} \colon u_i \pref v_j \pref X \pref \cdots.$
			
			\item For each $j \in [n]$, we have an \emph{$a$-voter} $a_j$ with 
			preference
			order
			$a_j \colon v_j \pref a \pref X \pref \cdots.$
			
			\item For each $i \in [n]$, we have a \emph{$b$-voter}~$b_i$ with 
			preference order
			$b_i \colon b \pref u_i \pref X \pref \cdots.$
			
			\item Finally, 
			the \emph{score voters} implement the desired Plurality
			scores. 
			For each candidate $c \in U(G) \cup V(G) \cup \{p\}$, there are
			exactly as many voters with preference order
			$c \pref X \pref \cdots$ as necessary to ensure that in total $c$
			has score $n$.
			Similarly, for each $x_i \in X$ there are $n$~voters with
			preference order $x_i \pref X \setminus \{x_i\} \pref \cdots$.
			There are also $n-k$~voters with preference order
			$a \pref X \pref \cdots$ and $k$ voters with preference order
			$b \pref X \pref
			\cdots$.  
		\end{enumerate}
		Let~$E$ be an election with the above-described candidates and
		voters. We form an instance $I$ of \textsc{Plurality \#Swap-Bribery}
		with this election, unit prices, and swap radius~$r := 3k$. Then, we
		make
		an 
		oracle query for~$I$ 
		and return its answer.
		In the remainder of the proof, we argue that this answer is equal to
		the number of size-$k$ matchings in~$G$.
		The 
		idea is that to make $p$ a winner, we have to transfer $k$
		points from $b$ to $a$ via swaps
		that correspond to a matching.
		
		Let $E'$ be some election in $R(E,r)$, i.e., an election at swap
		distance~$r$ from~$E$, where $p$~wins. We note that $p$ and the
		candidates from $X$ have score~$n$ in~$E'$ (indeed, in elections
		from $R(E,r)$, $p$ has score at most $n$ and the average score of
		the candidates in $X$ is at least $n$).
		Further,
		in $E'$ each edge voter, $a$-voter, and $b$-voter either
		ranks on top the same candidate  as in $E$, or the candidate that he or 
		she 
		ranked second in $E$ (otherwise some candidate in $X$ would
		have score above $n$).  We call  this the \emph{top-two rule}.
		
		Since $b$ must have at most $n$ points in $E'$, by the top-two rule,
		there must be at least $k$ $b$-voters that rank members of $U(G)$ on
		top. Let $U_b$ be the set of these members of $U(G)$. As each member
		of $U(G)$ can be swapped with $b$ at most once in the $b$-votes, we
		have $|U_b| \geq k$.
		
		Compared to $E$, in $E'$ each member of $U_b$ gets an additional
		point from the $b$-voters. Thus, for each $u_i \in U_b$ there must be
		a voter that ranked $u_i$ on top in $E$ but does not do so in $E'$.
		By the top-two rule, this must be an edge voter.  Let~$M$ be the set
		of pairs $\{u_i,v_j\}$ such that in~$E$ edge voter~$e_{ij}$
		ranks~$u_i$ on top, but in~$E'$ he or she ranks~$v_j$ on top.
		Naturally, we must have $|M| \geq |U_b|$.
		
		For each pair $\{u_i,v_j\} \in M$, there must be a voter who 
		swapped~$v_j$
		out of the top position in $E'$, because otherwise $v_j$~would have
		more than $n$~points. By similar arguments as before, this must be
		voter~$a_j$.  Let $V_a$ be the set of those members of~$V(G)$ that
		in~$E'$ are swapped out of the top positions in the $a$-votes.  It
		must be that
		$|V_a| \geq |M|$.
		
		Altogether, we have $|V_a| \geq |M| \geq |U_b| \geq k$ and, in fact,
		each of these sets must have exactly $k$ elements (because their
		elements correspond to unique swaps).
		Further, $M$ is a matching. If it were not, then some member of
		$U(G) \uplus V(G)$ would appear in two pairs in $M$, but then we
		would have to have two $a$-voters or two $b$-voters that
		corresponded to this candidate, which is not possible in our
		construction.
		
		This way we have shown that for each election in $R(E,r)$ where $p$
		wins, there is a corresponding size-$k$ matching. As the other
		direction is immediate, the proof is complete.
	\end{proof}
	
	A natural way to circumvent such intractability results is to seek
	$\fpt$ algorithms parameterized by the number of candidates or by the
	number of voters.
	For the former, 
	one typically expresses \textsc{Swap-Bribery} problems as 
	integer linear programs
	(ILPs)
	and invokes the classic algorithm of \citet{len:j:integer-fixed}, or
	some more recent one; see, e.g., the work of
	\citet{kno-kou-mni:j:fpt-swap-bribery}.  Unfortunately, counting
	analogs of these algorithms, dating back to the seminal work of
	\citet{bar:j:barvinok}, have $\xp$ running times and cannot be used
	for our purpose. Thus, we leave the complexity of our problems
	parameterized by the number of candidates open.
	Yet, for unit prices we do show an $\fpt$ algorithm parameterized by the 
	number of
	voters. 
	
	\begin{theorem}\label{thm:plurality-swap-bribery-fpt-n}
		For unit prices, \textsc{Plurality \#Swap-Bribery} parameterized by
		the number of voters is in $\fpt$.
	\end{theorem}
	\begin{proof}[Proof sketch]
		Consider an instance $I$ of \textsc{Plurality \#Swap-Bribery} with
		election $E = (C,V)$, where $C=\{c_1, c_2, \ldots, c_m\}$ and $V$
		contains $n$ voters. Let $r$ be the swap radius and, w.l.o.g., let
		$p=c_1$ be the designated candidate.
		
		The core idea is to go over all possible sequences
		$\mathcal{V} = (V_1, \ldots, V_{m'})$ such that
		(a)~$m'\leq \min(m,n)$,
		(b)~each $V_i$ is a subcollection of $V$ (consisting of not necessarily
		consecutive voters),
		(c)~each voter belongs to exactly one~$V_i$, and
		(d)~group $V_1$ has at least as many voters as every other group.
		For each such sequence, we solve the following \emph{global counting
			problem}: Count the number of ways to perform exactly~$r$ swaps so
		that
		(i) within each $V_i$, each two voters rank the same candidate,
		denoted~$c(V_i)$, on top,
		(ii) all voters in~$V_1$ rank $p$ on top (i.e., $c(V_1) = p$), and
		(iii) for each two groups $V_i$ and $V_{i'}$, where $i < i'$ and we
		have $c_j=c(V_i)$ and~$c_{j'}=c(V_{i'})$, it holds that~$j<j'$.
		These conditions ensure that after performing the swaps, each group
		votes for a different candidate, each candidate $c(V_i)$ receives
		exactly $|V_i|$ points, and $p$ wins. One can verify that every
		solution for our input instance corresponds to exactly one
		sequence~$\mathcal{V}$. In other words, to obtain the answer
		for~$I$, we need to sum up the answers for the global counting
		problems for each $\mathcal{V}$.
		
		To solve a given global counting problem in polynomial time,
		we define $T[i,\ell,s]$ to be the number of ways to perform
		exactly~$s$ swaps within the first~$i$ voter groups, so that
		conditions (i)--(iii) hold for $V_1, \ldots V_i$, and so that
		$c(V_i) = c_\ell$. We compute these values using dynamic programming
		(which requires solving a \emph{local counting problem}, also via
		dynamic programming, to count for each voter group the number of
		ways to ensure that all its members rank a given candidate on top).
		The solution for the global counting problem is then
		$\sum_{\ell \leq m}T[m',\ell,r]$.
		
		As the number of global counting problems to solve is bounded by a
		function of $n$, and each such problem is solved in polynomial time,
		the algorithm runs in $\fpt$ time with respect to the number of
		voters.
	\end{proof}
	
	The restriction to unit prices in
	Theorem~\ref{thm:plurality-swap-bribery-fpt-n} is
	necessary. Otherwise, a reduction from the problem of counting linear
	extensions of a partially ordered
	set by~\citet{bri-win:j:linear-extensions} shows $\sharpp$-hardness even
	for a single voter.
	
	\begin{restatable}{theorem}{thmprices}
		\label{thm:prices}
		\textsc{Plurality \#Swap-Bribery} is $\sharpp$-hard even for a
		single voter and unary-encoded prices.
	\end{restatable}
	
	We conclude with a brief mention of \textsc{\#Shift-Bribery}. Both the
	constructive and the destructive variant are in~$\p$, even with
	arbitrary unary-encoded prices (for the binary encoding,  
	$\sharpp$-hardness follows by a 
	reduction from 
	\textsc{\#Partition}). 
	Our algorithms use dynamic programming over groups 
	of voters with the same candidate
	as their top choice.
	
	\begin{restatable}{theorem}{thmpluralityshifts}
		\label{thm:plurality-shifts}
		For unary-encoded prices, both the constructive and the destructive
		variant of \textsc{Plurality \#Shift-Bribery} are in $\p$.
	\end{restatable}
	
	\subsection{Borda and \textsc{\#Shift-Bribery}}
	Our results for \textsc{Borda \#Swap-Bribery} follow from those for
	\textsc{\#Shift-Bribrey}, so we discuss the latter problem first.
	
	In the decision setting, the constructive variant of \textsc{Borda
		Shift-Bribery} is $\np$-hard (and is in $\fpt$ when parameterized by
	the shift radius, but is $\wone$-hard for the number of voters),
	whereas the destructive variant is in $\p$. In the counting
	setting, 
	both variants are $\sharpp$-hard and $\sharpwone$-hard for the
	parameterization by the number of voters; the result for the
	constructive case follows from a proof for the decision variant due to
	\citet{bre-fal-nie-tal:c:shift-bribery-committees} and for the
	destructive case, we use a similar approach with a few tricks on top.
	
	\begin{restatable}{theorem}{thmbordashiftvoters}
		\label{thm:borda-shift-voters}
		Both the constructive and the destructive variant of \textsc{Borda 
		\#Shift-Bribery} are $\sharpp$-hard and $\sharpwone$-hard when
		parameterized by the number of voters.
	\end{restatable}
	
	More surprisingly, for the parameterization by the shift radius, the
	constructive variant is in $\fpt$ and the destructive variant is
	$\sharpwone$-hard. Not only does the problem that was easier in the
	decision setting now became harder, but also---to the best of our
	knowledge---it is the first example where a destructive variant of an
	election-related problem is harder than the constructive one.
	Yet, 
	\textsc{Shift-Bribery} is quite special as the two variants differ
	both in the goal (i.e, whether we want the designated candidate to win
	or not) and in the available actions (shifting the designated
	candidate forward or backward; typically, destructive voting problems
	have the same sets of actions as the constructive ones).
	
	The $\fpt$ algorithm for the constructive case relies on the fact that
	if we can ensure victory of the designated candidate by shifting him
	or her by~$r$ positions forward, then there are at most $r$ candidates
	that we need to focus on (the others will be defeated irrespective
	what exact shifts we make).  There are no such bounds in the
	destructive setting.
	
	\begin{restatable}{theorem}{thmbordashiftradius}
		\label{thm:borda-shift-radius}
		Parameterized by the shift radius, \textsc{Borda \#Constructive
			Shift-Bribery} is in $\fpt$ (for unary-encoded prices), but the
		destructive variant is $\sharpwone$-hard, even for unit prices.
	\end{restatable}
	
	\begin{proof}[Proof (destructive case)]
		We give a polynomial-time Turing reduction from \textsc{\#Matching}
		to \textsc{Borda \#Destructive Shift-Bribery}. Let $(G,k)$ be an
		instance of \textsc{\#Matching}, where $G$ is a bipartite graph with
		vertex set $U(G) \uplus V(G)$, and $k$ is a positive
		integer. Without loss of generality, we assume that
		$U(G) = \{u_1, \ldots, u_n\}$, $V(G) = \{v_1, \ldots, v_n\}$, and
		$k \leq n$.
		
		Our reduction proceeds as follows. First, we form the set of
		relevant candidates $R := \{d,p\} \uplus U(G) \uplus V(G)$, where
		$d$~is the designated candidate. Moreover, for each relevant
		candidate $r\in R$, we form a set $D(r)$ of $3k+1$ dummy ones.  We
		will form an election $E$, where these candidates will have the
		following Borda scores ($X$ is some positive integer, whose value
		depends on the specifics of the construction; we will be counting
		ways in which $d$ can cease to be a winner by shifting him or her
		backward by $3k$ positions):
		\begin{align}
		\label{eq:1}  &\score_E(d) = X+3k,  \\
		\label{eq:2}  &\score_E(p) = X-k+1,\\
		\label{eq:3}  &\score_E(u_1) = \cdots = \score_E(u_n) = X-1, \\
		\label{eq:4}  &\score_E(v_1) = \cdots = \score_E(v_n) = X-1, \text{
			and}\\
		\label{eq:5}  &\text{each dummy candidate has score at most $X-3k-1$}. 
		\end{align}
		Election $E$ contains the following voters:
		\begin{enumerate}
			\item For each edge $e = \{u_i,v_j\}$ of the input graph, there
			is an \emph{edge voter} $v_e$ with preference order
			$v_e \colon d \pref u_i \pref v_j \pref p \pref D(p) \pref \cdots.$
			
			\item There is a group of voters who ensure that the scores are as
			described above. They rank at least $3k$ dummy candidates between
			each two relevant ones (so shifting $d$ by $3k$ positions back
			cannot change the score of another relevant candidate).

                        More precisely, we have a group of \emph{score
                          voters}, who ensure that
                        conditions~\eqref{eq:1}--\eqref{eq:5}
                        hold. Let $\pi$ be the following preference
                        order, giving the basic pattern for forming
                        all score voters:
			\begin{align*}
			\pi \colon & D(u_1) \pref u_1 \pref \cdots  \pref D(u_n) \pref u_n 
			\pref D(p) \pref p \pref \\
			& D(v_1) \pref v_1 \pref \cdots \pref D(v_n) \pref v_n \pref D(d) 
			\pref d.
			\end{align*}
			The crucial feature of $\pi$ is that relevant candidates are
			separated from each other with at least $3k+1$ dummy ones; their
			particular order is not important.  For each relevant candidate
			$c \in R$, let $\increase(c)$ be a pair of preference orders where
			one is equal to reversed~$\pi$, and the other is identical
			to~$\pi$ except that $c$ is shifted one position forward. By
			adding a pair of voters with preference orders $\increase(c)$ to
			the election, we increase the score of~$c$ by~$|C|$, 
			the scores of the other relevant candidates by $|C|-1$,
			and 
			the scores of 
			the dummy candidates by at most $|C|-1$. For each relevant
			candidate $c \in R$, we add polynomially many such pairs of
			voters---the polynomial is with respect to $n$ and~$k$---as follows:
                        \begin{enumerate}
                        \item
                        For each relevant candidate $c \in R$
			we add the smallest number of pairs of voters $\increase(c)$ so
			that, taken together with the edge voters, all relevant candidates
			have identical scores and each relevant candidate has a higher
			score than each dummy candidate.
                      \item
                        For each relevant
			candidate $c \in R$, we add $3k+1$ pairs of voters $\increase(c)$
			(this ensures that, in total, each relevant candidate has at least
			$3k$ points more than each dummy candidate).
                      \item
                        We add
			$4k$~pairs of voters $\increase(d)$, one pair of voters
			$\increase(p)$, and for each $i \in [n]$, we add $k-1$~pairs of
			voters $\increase(u_i)$ and $k-1$ pairs of voters $\increase(v_i)$
			(this ensures that the scores of all relevant candidates are as
			promised).
                        \end{enumerate}
                        This ensures that the scores satsify conditions~\eqref{eq:1}--\eqref{eq:5}.
		\end{enumerate}
		Next, we form an election~$F$ identical to~$E$, except that one of
		the edge voters ranks $p$ one position lower
		(so that $p$'s score
		in~$F$ is $X-k$).  Let $I_E$ and $I_F$ be instances of \textsc{Borda
			\#Destructive Shift-Bribery} with designated candidate $d$, shift
		radius~$3k$, unit prices, and elections $E$ and~$F$,
		respectively. Our reduction queries the oracle for the numbers of
		solutions for $I_E$ and~$I_F$, subtracts the latter from the former,
		and outputs this value. We claim that it is exactly the number of
		size-$k$ matchings in~$G$.
		
		To see why this is the case, consider some solution for~$I_E$. There
		are two possibilities: Either $d$ passes some member of
		$U(G) \uplus V(G)$ twice (in which case this candidate gets at least
		$X+1$~points, whereas $d$~always gets exactly $X$~points), or
		$d$~passes each member of $U(G) \uplus V(G)$ at most once. In the
		latter case, only~$p$ can defeat~$d$ (all the other candidates have
		at most $X$ points). However, for this to happen, $d$ must pass $p$
		exactly $k$ times (with the shift radius of $3k$, $d$ cannot pass
		$p$ more times). Further, since we assumed that $d$ never passes a
		member of $U(G) \uplus V(G)$ more than once, the votes where $d$
		passes $p$ must correspond to a size-$k$ matching in~$G$.  We refer
		to such solutions as \emph{matching solutions}.
		
		The set of solutions for~$I_F$ contains all the solutions for~$I_E$
		except for the matching ones (because in $I_F$, $p$~ends up with
		$X$~points and not $X+1$). So, by subtracting the number of
		solutions for~$I_F$ from the number of solutions for~$I_E$, we get
		exactly the number of size-$k$ matchings in~$G$.
	\end{proof}

	For \textsc{Borda \#Swap-Bribery}, we obtain
        $\sharpp$-hardness and $\sharpwone$-hardness for the
        parameterization by the number of voters by noting that the
        proofs for \textsc{\#Shift-Bribery} still apply in this case
        (regarding $\sharpp$-hardness,
        \citet{bau-hog:c:counting-swap-bribery-evaluation} also report
        a related result). The parameterization by the swap radius
        remains open, though (the proof of
        Theorem~\ref{thm:borda-shift-radius} does not work as many
        new, hard to control, solutions appear).
	
	\begin{corollary}
		\textsc{Borda \#Swap-Bribery} is $\sharpp$-hard and
		$\sharpwone$-hard for the parameterization by the number of voters,
		even for the case of unit prices.
	\end{corollary}
	
	\section{Experiments}\label{se:exp}
	In the following, we use \textsc{\#Swap-Bribery} to analyze the
	robustness of election winners experimentally.  For clarity, in this
	section we use normalized swap distances, which specify the fraction
	of all possible swaps in a given election.\bigskip
	
	\noindent\textbf{Setup.}\quad
	We used a dataset of
	$800$ elections, each with $10$ candidates and $100$ voters, prepared
	by \citet{szu-fal-sko-sli-tal:c:map-of-elections}; in the appendix we
	also show results for different election sizes.\footnote{%
		Datasets with various numbers of candidates are available at
		Szufa et al.'s website \texttt{https://mapel.readthedocs.io/}.}
	This dataset contains
	elections generated from various statistical cultures, of which for us
	the most relevant are the following ones (we provide intuitions only;
	for details we point, e.g., to the work of
	\citet{szu-fal-sko-sli-tal:c:map-of-elections}):
	\begin{enumerate}
		\item The impartial culture model (IC), where each election consists
		of preference orders chosen uniformly at random.
		
		\item The urn model, with the parameter of contagion
		$\alpha \in \mathbb{R_+}$, where for $\alpha=0$ the model is
		equivalent to IC, but as $\alpha$~grows, larger and larger  groups of 
		identical
		votes become more probable.
		
		\item The Mallows model, with dispersion parameter $\phi \in [0,1]$,
		where the votes are generated by perturbing a given central one; for
		$\phi = 0$ only the central vote appears, and for
		$\phi = 1$ the model is equivalent to IC.
		
		\item The $t$D-Cube/Sphere 
		models, where the candidates and voters are
		points 
		in a $t$-dimensional hypercube/sphere and the voters rank the
		candidates by distance (we refer to 1D/2D-Cube elections as
		1D-Interval/2D-Square ones).
		
	\end{enumerate} 
	
	\noindent
	\citet{szu-fal-sko-sli-tal:c:map-of-elections} present their elections
	as a \emph{map} (see \Cref{fig:mapmap}). Note that the map contains
	elections from a number of distributions beyond those mentioned above
	(for details we point to their original work or to the appendix); ``IC
	and similar'' refers to IC elections, Mallows elections with
	$\phi$~values close to~$1$, and a few other elections that are similar
	to IC under the metric of
	\citet{szu-fal-sko-sli-tal:c:map-of-elections}. Later we will use the
	map to present our results and to identify some patterns.  \bigskip

	\begin{figure*}[t]
		\centering
			\includegraphics[width=0.55\textwidth]{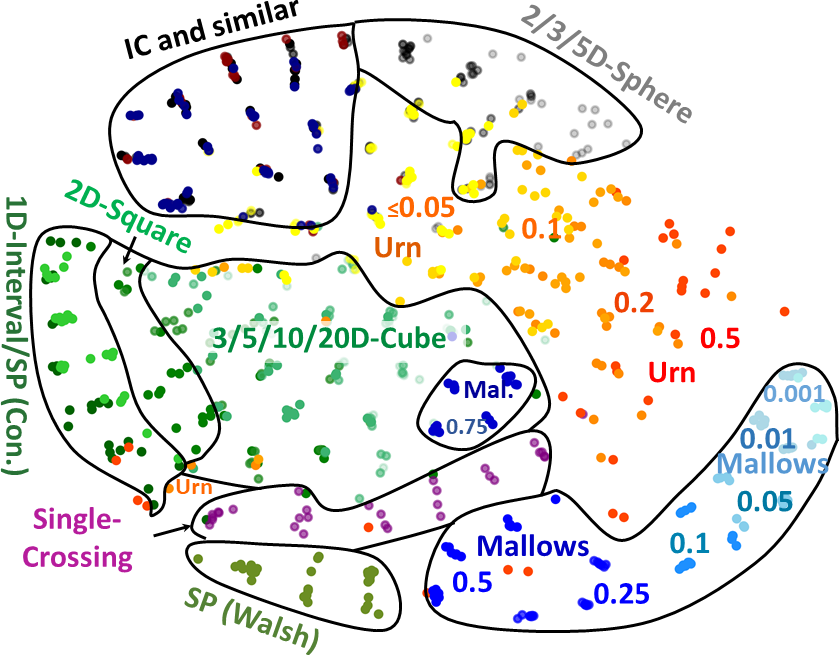}
			\caption{The map of elections, due to
				\citet{szu-fal-sko-sli-tal:c:map-of-elections}. Each point
				corresponds to an election and its color gives the model from
				which it came. Generally, the closer two points are, the more
				similar are the corresponding elections in their metric.}
			\label{fig:mapmap}
                    \end{figure*}

       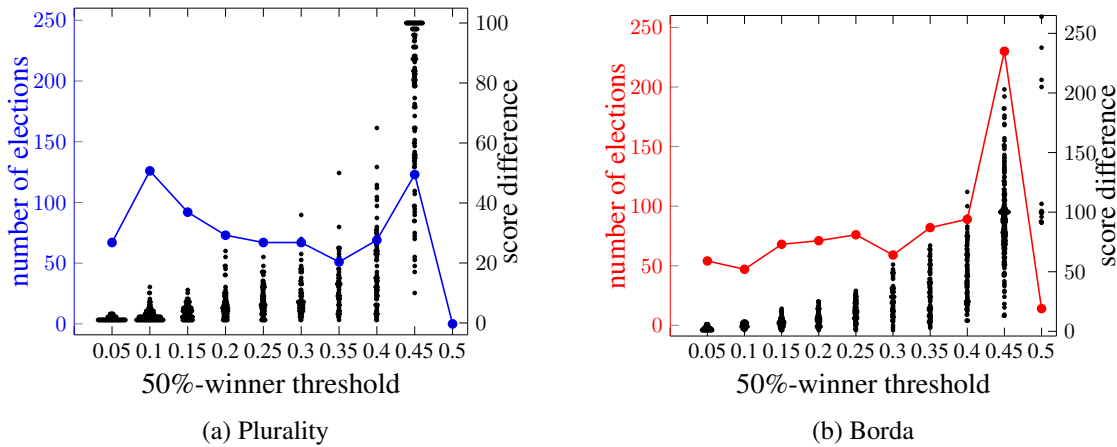
\begin{figure*}
			\centering 
			\begin{subfigure}[t]{0.47\textwidth}
				\centering            
				\resizebox{\textwidth}{!}{\input{scatterplur.tex}}
				\caption{Plurality}
				\label{fig:distplur}
			\end{subfigure}
			\hfill
			\begin{subfigure}[t]{0.47\textwidth}
				\centering
				\resizebox{\textwidth}{!}{\input{scatterborda.tex}}
				\caption{Borda}
				\label{fig:distborda}
			\end{subfigure}
			\caption{Blue and red solid lines show the number of elections
				(on the $y$-axis) with a given $50\%$-winner threshold (on the
				$x$-axis) for Plurality and Borda, respectively.  In the black
				scatter plots, each dot corresponds to an election. Its 
				$x$-coordinate 
				gives the $50\%$-winner threshold
				(perturbed, if many elections would overlap) and its 
				$y$-coordinate 
				gives the difference between the scores of the
				winner and the runner-up.}
			\label{fig:dist}
	\end{figure*}

	\noindent\textbf{Computations.}\quad
	For each election $E$ and candidate $c$, let $P_{E,c}(r)$ be the
	probability that $c$ wins---under a given voting rule---in an election
	chosen uniformly at random from $R(E,r)$.  Ideally, we would like to
	compute these values for all elections, candidates and swap distances,
	for Plurality and Borda.  However, since \textsc{\#Swap-Bribery} is
	$\sharpp$-hard for both our rules, instead of computing these values
	exactly, we resorted to sampling.
	Specifically, for each election (except those with tied winners) and
	each normalized swap distance $r \in \{0.05, 0.1, \ldots, 1\}$ we
	sampled $500$ elections at this distance and for each candidate
	recorded the proportion of elections where he or she won\footnote{By
		Hoeffding's inequality, the probability that the
		estimated winning probability for a given candidate deviates by more
		than $0.1$ from the true one can be upper bounded by $0.1\%$.} (see 
	Appendix B.2
	for the sampling procedure).
	For each election, we quantified the {robustness} of its winner by
	identifying the smallest swap distance $r$, among the considered ones,
	for which he or she has winning probability below $50\%$. We refer to
	this value as the \emph{$50\%$-winner threshold} (or, threshold, for
	short).
	
	\bigskip
	
	\noindent\textbf{Results.}\quad
	In the following, we present several findings from our experiments,
	each followed by supporting arguments.
	\begin{conclusion}\label{f1}
		The Borda winner of an election is usually more robust against random 
		swaps than the Plurality winner. 
	\end{conclusion}

	The solid blue line in \Cref{fig:distplur} and the solid red line in
	\Cref{fig:distborda} show how many elections from our dataset have
	particular 50\%-winner thresholds for Plurality and Borda,
	respectively.  While for Borda the threshold of $0.45$ occurs far more
	often than the other ones, for Plurality, the distribution is more
	uniform (with small spikes at $0.1$ and $0.45$).  So, Plurality
	elections are more likely to change results after relatively few swaps
	than the Borda ones. Two 
	explanations are that (a)~under Plurality there can be ``strong
	contenders'' who do not win, but who are often ranked close to the
	first place and, thus, can overtake the original winner after a few
	swaps, and (b)~the Plurality winner has the highest chance of losing
	points, as he or she is ranked first most frequently.
	Under Borda, the candidates usually have similar chances of both
	gaining and losing a point with a single swap.
	
	\begin{conclusion}\label{fin:2}
		The score difference between the winner and the runner-up (i.e., the
		candidate ranked in the second place) has a limitied predictive
		value for the $50\%$-winner threshold.
	\end{conclusion}
	
	Let us consider the black scatter plots in Figures~\ref{fig:distplur}
	(for Plurality) and~\ref{fig:distborda} (for Borda).  There, each
	election is represented as a dot, whose $x$-coordinate is the
	$50\%$-winner threshold (perturbed a bit if many elections were to
	take the same place) and whose $y$-coordinate is the score difference
	between the winner and the runner up. While there certainly is some
	correlation betwen these two values, the same score difference may
	lead to a wide range of $50\%$-winner thresholds (e.g., for Plurality
	a score difference of $10$ may lead to the threshold being anything
	between $0.1$ and $0.4$).
	
	From now on, we focus on Plurality, but most of our conclusions also
	apply to Borda (we do mention some differences though; for details, see
	Appendix B.3).
	In \Cref{fig:mapplur}, we show the map of elections, with colors
	corresponding to each election's $50\%$-winner threshold. The figure
	also includes six plots, each showing the values of $P_{E,c}(r)$ for
	four candidates in six selected elections (we discuss them later).
	
	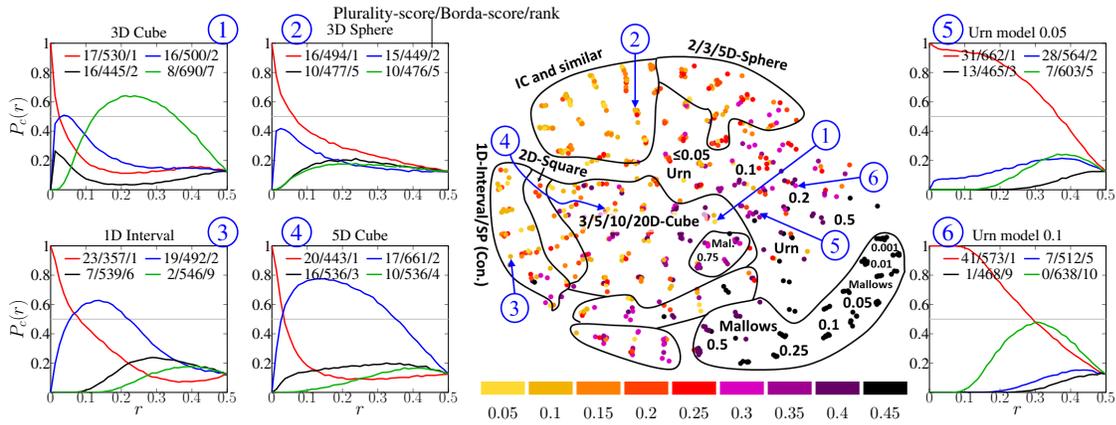
\begin{figure*}[t]
		\begin{center}
                  \resizebox{\textwidth}{!}{\input{l_image}}
		\end{center}
		\caption{Map of elections visualizing the $50\%$-winner threshold
			for Plurality and six plots showing $P_{E,c}(r)$ as a
			functon of $r$, for six selected elections and the four most
			successful candidates in each (see the paragraph preceding
			\Cref{con:winners-winning} for
			details). }\label{fig:mapplur}
	\end{figure*}
	
	\begin{conclusion}\label{f3}
		Positions of elections 
		on the map correlate with their $50\%$-winner thresholds. Elections
		sampled from the same model tend to have similar thresholds.
	\end{conclusion}
	Consider the map in Figure~\ref{fig:mapplur}. As we move from top to
	bottom and from left to right, the $50\%$-winner threshold tends to
	increase. Not surprisingly, it is low for IC elections (as they are
	completely random, it is natural that few changes can affect the
	result) and it is high for Mallows elections with low $\phi$ (most
	preference orders in these elections are identical, up to a few
	swaps).
	For urn elections, the threshold tends to increase with parameter
	$\alpha$ (as the votes become less varied with larger $\alpha$).
	Interestingly, the threshold is somewhat more varied among $t$D-Cube
	elections (as compared to the other models), and one can notice that
	for 1D-Interval elections it tends to be slightly lower than for
	higher-dimensional $t$D-Cube ones (this effect is much stronger for
	Borda).  Yet, typically elections generated from a given model (with a
	given parameter) tend to have similar threshold values.
	
	For further insights, we turn to the six plots in
	\Cref{fig:mapplur}. 
	Each of them regards a particular election and four of its candidates.
	The candidates are marked with colors and the original winner is
	always red. 
	For each considered election $E$ and each candidate $c$ in the plot,
	we show $P_{E,c}(r)$ for values of $r$ between $0$ and $0.5$
	(specifically, for these six elections, we estimated $P_{E,c}(r)$ for
	$r \in \{0.0125, 0.025, \ldots, 0.5\}$ using 10'000 samples in each
	case).  We limited the range of relative swap distances
	because
	above $0.5$, the votes are
	becoming similar to the reverses of the original ones.
	For each of the candidates, in the legend we provide his or her
	Plurality score, Borda score, and the rank in the original election
	(we use the Borda scores in further discussions).
	For each of the six elections, we sorted the candidates with respect
	to $\max_{r \in \{0, 0.0125, \ldots, 0.5\}}P_{E,c}(r)$ and chose the
	top four to be included in the plot.
	The elections were chosen to show interesting phenomena (thus the
	patterns they illustrate are not always the most common ones, but are
	not outliers either).  The following discussion refines the
	observations from \Cref{fin:2}.

	\begin{conclusion}\label{con:winners-winning}
		Winners winning by a small margin are not necessarily close to
		losing. Winners winning by a large margin are robust but not
		necessarily very robust winners.
	\end{conclusion}
	
	In Elections 1 to 4, the winners are very sensitive to random
	swaps: The \second/ candidate already wins a considerable proportion
	of elections even if only a $0.0125$ fraction of possible swaps are
	applied (i.e., about half a swap per vote, on average), and the
	\first/ candidate quickly drops below $50\%$ winning probability.  It
	is quite surprising that so few random swaps may change the outcome
	with fairly high probability.
	There are also differences among these four elections. For example, in
	Elections~1 and~2 the candidates have similar scores, but in Election
	2 the \first/ candidate stays the most probable winner until swap
	distance $0.4$, whereas in Election~1, the most probable winner
	changes quite early.  The plots for Elections~3 and~4 are similar to
	that for Election~1, but come from tD-Cube elections of different
	dimension; this pattern appears in elections from other families of
	distributions too, but less commonly.
	
	In Elections~1 to 4, the original winner has at most four Plurality
	points of advantage over the next candidate, so one could argue that
	scores suffice to identify close elections.  Yet, in Election~5 the
	difference between the scores of the winner and the runner-up is $3$,
	but the \first/ candidate stays a winner with probability greater than
	$50\%$ until swap distance $0.35$. Thus, looking only at the scores can
	be misleading. Nonetheless, if the score difference is large (say,
	above $25$), the $50\%$-winner threshold is always above $0.2$ (see
	\Cref{fig:distplur}). But, as witnessed in Election 6, even in such
	seemingly clear elections, around $10\%$ of random swaps suffice to
	change the outcome with a non-negligible probability.
	
	\begin{conclusion}
		The score of a non-winning candidate has a limited predictive value
		for his or her probability of winning if some random swaps are
		performed.
	\end{conclusion}
	
	Perhaps surprisingly, in some elections the most probable winner at
	some (moderately low) swap distance is not necessarily ranked highly
	in the original election. For instance, in Election~1 the \fourth/
	candidate is originally ranked seventh, but becomes the most probable
	winner already around swap distance $0.1$. Here, this can be explained
	by the fact that he or she has a significantly higher Borda score
	than the other candidates. So, he or she is ranked highly in many
	votes and can reach the top positions with only a few swaps.
	Yet, not all patterns can be explained this way.
	For example, in Elections~1 and~2, the first two
	candidates have similar Plurality and Borda scores but still behave
	quite differently, even at small swap distances.
	
	\section{Conclusions}
	We have shown that the counting variants of \textsc{Swap-Bribery} have
	high worst-case complexity, but, nonetheless, are very useful for
	analyzing the robustness of elections winners. In particular, we have
	observed some interesting phenomena, including the fact that the
	scores of the candidates do not suffice to evaluate their strengths.
	Establishing the complexity of \textsc{Borda \#Swap-Bribery}
	parameterized by the swap radius remains as an intriguing open
	problem.
	
	\bigskip
	
	\noindent
	\textbf{Acknowledgments.}\quad 
	Niclas Boehmer was supported by the DFG project MaMu (NI 369/19).
	Piotr Faliszewski was supported by a Friedrich Wilhelm Bessel Award from 
	the Alexander von Humboldt Foundation.
	Work started while all authors were with TU~Berlin.

        \bibliography{bib}
	
	\appendix
	
	\section{Missing Proofs from Section~\ref{sec:complexity}}
	
	In this section, we provide missing details and proofs from
	Section~\ref{sec:complexity}.
	
	\subsection{Auxilary Algorithms}
	\label{app:vgc}
	In this section we provide a number of polynomial-time algorithms for
	solving problems of the following form: Given an election and a
	particular budget (or, number of swaps) compute the number of ways of
	performing exactly this many swaps so that the election has some given
	shape (e.g., all the voters ranks the same given candidate on top).
	We refer to such problems as \emph{voter group contribution counting
		problems}.
	
	\subsubsection{Swap Contribution for Plurality (Unit Prices)}
	Given an election~$E=(C,V)$, a budget~$r$, and a distinguished
	candidate~$p \in C$, $\vgcSwapPlu(V,r,p)$ denotes the number of
	possibilities to perform exactly~$r$ swaps so that~$p$ is the top
	choice of every voter within~$V$.
	
	\begin{lemma}\label{lem:vgcSwapPlu-P}
		One can compute $\vgcSwapPlu(V,r,p)$ in time $O(n\cdot m^4)$.
	\end{lemma}
	\begin{proof}
		Let $C=\{p,c_2,\ldots,c_m\}$, $V=\{v_1,\ldots,v_n\}$, $r$, and $p$
		be given as described above.  We define the following dynamic
		programming table~$L$.  An entry $L[i,r']$ denotes the number of
		possibilities to perform exactly~$r'$ swaps within the first
		$i$~votes in~$V$ so that~$p$ is the top choice for these~$i$ voters.
		
		Let $r^*(v)$ 
		denote the number of swaps required to push
		candidate~$p$ to the top position in vote
		$v$.  We initialize the table via:
		\begin{align*}
		L[1,r']=\begin{cases}
		0 & \text{if $r'<r^*(v_1)$} \\
		Y(m-1,r'-r^*(v_1)) & \text{otherwise},
		\end{cases}
		\end{align*}
		where
		$Y(m,k)$~denotes the number of permutations of swap
		distance~$k$ from a given permutation with
		$m$~elements (algorithms for computing this value in polynomial time
		are well known).  We update the table with increasing~$i$ via
		\begin{align*}
		L[i,r'] = \!\!\!\!\!\!\!\! \sum_{r^*(v_i) \le r'' \le 
		r'}\!\!\!\!\!\!\!\! Y(m-1,r''-r^*(v_i)) \cdot L[i-1,r'-r'']. 
		\end{align*}
		Finally, $\vgcSwapPlu(V,r,p)=L[n,r]$ gives the solution.
		
		The initialization is correct, because we have to push
		candidate~$p$ to the top position at
		cost~$r^*(v_1)$.  This fixes the first position and the other
		$m-1$~positions can be freely rearranged.  Naturally, there are
		$Y(m-1,r'-r^*(v_1))$ possibilities to do this.  Similarly, in the
		update step we sum over all possibilities to distribute our
		$r'$~swaps among the first~$i-1$ votes and the
		$i$th vote (again, we need at least $r^*(v_i)$ swaps to push
		$p$~to the top).  In each case, the number of possibilities is the
		product of all possibilities to spend~$r''$ for
		voter~$i$ and $r'-r''$~swaps for the first $i-1$~voters.
		
		The table is of dimension~$O(n\cdot m^2)$.  Computing each single
		table entry can be done in time $O(m^2)$: Precomputing the table~$Y$
		with all entries takes $O(m^2)$ time (see also Appendix
		\ref{app:samp} where a recurrence is given) and with this being
		done, computing each single table entry of~$L$ takes $O(m^2)$ time
		since there are at most~$m^2$ values for~$r''$ and for each such
		value we have to do only a constant number of arithmetic operations.
	\end{proof}
	
	\subsubsection{Shift Contribution for Plurality (Arbitrary Prices Encoded 
	in 
		Unary)}
	At first, let us consider the constructive variant of 
	\textsc{Shift-Bribery},
	where we can shift forward the preferred candidate $p$.
	
	Let~$E=(C,V)$ be an election where every voter prefers the same
	candidate~$d \in C$.  Let $r$~be a budget, $p \in C$~be a
	distinguished candidate, and $s$~be an integer score value.  Moreover,
	we are given some cost
	function~$c \colon (V,\mathbb{N}) \rightarrow \mathbb{N}$ describing
	the costs~$c(v,\ell)$ of shifting~$p$ by~$\ell$ positions forward.  We
	define $\vgcShifPlu(V,r,p,s)$ as the number of possibilities to shift
	candidate~$p$ forward at total costs~$r$ within~$V$ so that $p$~is
	ranked at the top position exactly $s$~times.
	
	\begin{lemma}\label{lem:vgcShiftPlu-P}
		One can compute $\vgcShifPlu(V,r,p,s)$ in time $O(n^2\cdot r^2)$.
	\end{lemma}
	\begin{proof}
		We assume that $V=\{v_1,\ldots,v_n\}$ and compute \vgcShifPlu~using
		standard dynamic programming.
		Let $L[j, s', r']$ be the number of ways to shift~$p$ forward in the
		first~$j$ votes in~$V$ at total cost of~$r'$, so that $p$~obtains
		$s'$~additional Plurality points.
		
		We introduce two auxilliary functions, \topSh{v}{x} and
		\valSh{v}{x}.  Function \valSh{v}{x} indicates whether spending
		cost~$x$ for voter~$v$ is valid (e.g., we cannot spend more than
		necessary to push~$p$ to the top position)
		and function \topSh{v}{x} indicates
		whether spending costs~$x$ for voter~$v$ is successful (i.e., pushes
		$p$~to the top position). Formally, we have:
		\begin{align*}
		\valSh{v}{x}&=\begin{cases}
		1 & \text{it is valid to spend cost~$x$ at voter~$v$,}\\
		0 & \text{otherwise.}
		\end{cases}\\
		\topSh{v}{x}&=\begin{cases}
		1 & \text{$p$ becomes top choice in~$v$ at cost~$x$,}\\
		0 & \text{otherwise.}
		\end{cases}
		\end{align*}
		We initialize our table with:
		$$L[1, s', r']=\begin{cases}
		\valSh{v_1}{r'} \cdot \topSh{v_1}{r'}      & s'=1\\
		\valSh{v_1}{r'} \cdot (1-\topSh{v_1}{r'})  & s'=0\\
		0                     & \text{otherwise.}
		\end{cases}\phantom{\text{123}}$$\\
		We update $L[j, s', r']$ for~$j>1$ via $L[j, s', r']=$
		\begin{align*}
		\sum_{r''\le r'} \big(
		L[j-1, s'-1, r'']\cdot \topSh{v_j}{r'-r''} \cdot \valSh{v_j}{r'-r''}\\ +
		L[j-1, s', r'']\cdot (1-\topSh{v_j}{r'-r''}) \cdot 
		\valSh{v_j}{r'-r''}\big).
		\end{align*}
		The table~$L$ is of size~$n^2\cdot r$.
		Computing a single table entry requires at most $2r$~table lookups
		and at most $2r$~arithmetic operations.
	\end{proof}
	
	Let us now consider the destructive case, where we can push a given
	candidate backward.  Consider an election~$E=(C,V)$ with two
	distinguished candidates, $p$ and~$d$, such that every voter either
	prefers~$p$ the most while ranking candidate~$d$ in the second
	position, or prefers~$d$ the most.  Let $r$~be the budget, and $s$~be
	an integer score value.  Moreover, we are given some cost
	function~$c \colon (V,\mathbb{N}) \rightarrow \mathbb{N}$ describing the
	costs~$c(v,\ell)$ of shifting~$p$ by~$\ell$ positions backward.  We
	define $\vgcDShifPlu(V,r,p,s)$ as the number of possibilities to shift
	candidate~$p$ backward at total costs~$r$ within~$V$ such that $p$~is
	ranked exactly $s$~times at the top position.
	
	\begin{lemma}\label{lem:vgcDShiftPlu-P}
		One can compute $\vgcDShifPlu(V,r,p,s)$ in time $O(n^2\cdot r^3)$.
	\end{lemma}
	\begin{proof}
		We assume that $V=\{v_1,\ldots,v_n\}$ and compute \vgcDShifPlu using
		yet again standard dynamic programming, defining our table~$L$ as
		follows.  An entry $L[j, s', r']$ contains the number of ways to
		shift~$p$ backward in the first~$j$ votes from~$V$ at total cost
		of~$r'$, so that $p$~ends up with exactly $s'$~points.
		
		As in the constructive case, we introduce two auxilliary functions,
		\topSh{v}{x} and \valSh{v}{x}.  Function \valSh{v}{x} indicates that
		spending cost~$x$ for voter~$v$ is valid
		and function \topSh{v}{x} indicates that spending cost~$x$ for
		voter~$v$ is successful (pushes $d$~to the top position). Formally:
		\begin{align*}
		\valSh{v}{x}&=\begin{cases}
		1 & \text{it is valid to spend cost~$x$ at voter~$v$}\\
		0 & \text{otherwise.}
		\end{cases} \\
		\topSh{v}{x}&=\begin{cases}
		1 & \text{$d$ becomes top choice in~$v$ at cost~$x$}\\
		0 & \text{otherwise.}
		\end{cases}
		\end{align*}
		We initialize our table with:
		$$L[1, s', r']=\begin{cases}
		\valSh{v_1}{r'} \cdot (1-\topSh{v_1}{r'}) & s'=1\\
		\valSh{v_1}{r'} \cdot \topSh{v_1}{r'}     & s'=0\\
		0                 & \text{otherwise.}
		\end{cases}\phantom{\text{12}}$$
		
		\noindent We update $L[j, s', r']$ for~$j>1$ via $L[j, s', r']=$
		\begin{align*}
		\textstyle \sum_{r''\le r'} \big(
		L[j-1, s', r''] \topSh{v_j}{r'-r''}  \valSh{v_j}{r'-r''}\\ +
		L[j-1, s'-1, r''] (1-\topSh{v_j}{r'-r''})  \valSh{v_j}{r'-r''}\big).
		\end{align*}
		
		The table~$L$ is of size~$n^2\cdot r$.
		Computing a single table entry requires at most $2r$~table lookups
		and at most $2r$~arithmetic operations.
	\end{proof}
	
	\subsection{Missing Details for the Proof of 
	Theorem~\ref{thm:plurality-swap-bribery-fpt-n}}
	\label{app:thm:plurality-swap-bribery-fpt-n}
	
	For the proof of Theorem~\ref{thm:plurality-swap-bribery-fpt-n}, we need to
	discuss how to compute the table~$T$.
	
	\begin{lemma}
		Table~$T$ can be computed in time $O(n\cdot m^7)$.
	\end{lemma}
	\begin{proof}
		We initialize the table by setting $T[1,1,r']=\vgcSwapPlu(V_1,r',p)$
		and $T[1,\ell,r']=0$, $\forall \ell>1$.  The table is filled with
		increasing~$i$ by setting $T[i,\ell,r']$ to be:
		\[
		\sum_{\ell'<\ell, r''\le r'} (T[i-1,\ell',r''] \cdot
		\vgcSwapPlu(V_i,r'-r'',c_\ell)).
		\]
		
		Note that the initialization is correct by the definition
		of~$\vgcSwapPlu(V,r,c)$ and the fact that the first group must
		consistently vote for~$c_1=p$.  For updating the table, we sum up
		over all possibilities to split the swap budget~$r'$ between the
		$i$th voter group and first~$i-1$ voter groups, in combination with
		each possible candidate~$c_{\ell'}$ that may have been pushed to the
		top position by all voters of group~$i-1$.
		
		Computing table~$T$ requires filling in $m \cdot m \cdot nm^2$ entries, 
		and
		each entry takes $O(m^3)$ time (due to the number of terms in the
		sum in the update step). All in all, the algorithm takes time
		$O(n\cdot m^7)$, assuming that each of at most~$m$
		functions~$\vgcSwapPlu$ was computed in time $O(n\cdot m^4)$.
	\end{proof}
	
	\subsection{Proof of Theorem~\ref{thm:prices}}
	
	\thmprices*
	\begin{proof}
		In an instance of \textsc{\#Linear Extensions} we are given a set
		$Z = \{z_1, \ldots, z_m\}$ of items and a set
		$O \subseteq X \times X$ of constraints; we ask for the number of
		linear orders over $Z$ such that for each constraint $(x,y) \in O$,
		$x$ precedes~$y$. We reduce this problem to \textsc{Plurality 
		\#Swap-Bribery} with $0/1$ prices (i.e., each swap either has a unit 
		cost
		or is free) and budget $r := 0$ (if one preferred to avoid zero
		prices, then doing so would require only a few adaptations in the
		proof).
		
		Given an instance of \textsc{\#Linear Extensions}, as specified
		above, first we compute a single order $\pref$ that is consistent
		with the constraints (doing so is easy via standard topological
		sorting; if no such order exists, then we return zero and
		terminate).  Next, we form an election $E$ with candidate set
		$C = \{p\} \uplus Z$ and a single vote $v$, where $p$ is ranked
		first and all the other candidates are ranked below, in the order
		provided by $\pref$. We set the swap prices so that:
		\begin{enumerate}
			\item For each candidate $z \in Z$, the price for swapping him or
			her with~$p$ is one.
			\item For each pair $(x,y) \in O$, the price for swapping $x$
			and $y$ is one.
			\item All other prices are zero.
		\end{enumerate}
		We form an instance of \textsc{Plurality \#Swap-Bribery} with this
		election, prices, and budget $r := 0$. We make a single query
		regarding this instance and output the obtained value.
		
		To see that the reduction is correct, we notice that for every
		preference order that~$v$ may have after performing swaps of price
		zero, it holds that (a)~$p$ is ranked first (because swapping~$p$
		out of the first position has nonzero price) and (b)~for all pairs
		$(x,y) \in O$, $x$ is ranked ahead of $y$ (because swapping $x$ and
		$y$ has nonzero cost). On the contrary, for every linear order
		$\pref'$ that is consistent with $O$, it is possible to transform
		the preference order of $v$ so that $p$ is ranked first, followed by
		members of $Z$ in the order specified by $\pref'$ (for each two
		candidates $x,y \in Z$ such that $x \pref y$ but $y \pref' x$, the
		cost of swapping them is zero, and if we have not transformed our
		vote into the desired form yet, then there are always two such 
		candidates
		that are ranked consecutively).
	\end{proof}
	
	\subsection{Proof of Theorem~\ref{thm:plurality-shifts}} 
	\label{ap:plur-shift}
	
	\thmpluralityshifts*
	
	\begin{proof}[Proof for the constructive case]
		A core observation for our algorithm is that, under
		\textsc{\#Shift-Bribery}, every voter will either vote for~$p$ (if
		we shift~$p$ to the top position) or for its original top choice.
		This allows us to group the voters according to their top choices as
		follows.  Let $(V_0, V_1,V_2,\ldots,V_{m'})$, where $m'\le m-1$, be
		a partition of voters into groups so that (in the original election)
		every two voters within each group~$V_i$ share the same top choice,
		while every two voters from different groups have different top
		choices; additionally, we require that the voters in group $V_0$
		rank the distinguished candidate~$p$ on the top position.
		
		For each candidate $c$, let $\scoreof{c}$ be the original score
		of~$c$.  The idea of our algorithm is to count, for each possible
		final score~$\finScP \ge \scoreof{p}$ of~$p$, the number of ways to
		spend the given budget~$r$, so that $p$ obtains $\finScP$ points
		while no other candidate obtains more than~$\finScP$ points.  For
		each possible $\finScP$ we create one global dynamic programming
		table~$T_{\finScP}$.
		
		Our global tables are defined as follows.  An entry
		$T_{\finScP}[i, s', r']$ contains the number of ways to shift~$p$
		forward in the voter groups $V_1,\ldots,V_i$ at total cost of~$r'$,
		so that $p$~obtains $s'$~additional points while no top choice from
		any voter group $V_1,\ldots,V_i$ receives more than~$\finScP$
		points.
		
		To compute the values in table~$T$, we will use the following
		\emph{local counting problem}, maintaining the \emph{voter group
			contribution}: Given a voter group~$V'$, a budget~$r'$, a
		distinguished candidate~$p$, and a score~$s'$, compute the number of
		possibilities $\vgcShifPlu(V',r',p,s')$ to shift candidate~$p$ by in
		total~$r'$ positions within~$V'$, so that $p$~is ranked exactly
		$s'$~times at the top position.  This number can be computed in
		polynomial time (see Lemma~\ref{lem:vgcShiftPlu-P} in
		Appendix~\ref{app:vgc}).
		
		The intitialization of $T_{\finScP}$ is straight-forward, by setting:
		\[
		T_{\finScP}[1, s', r']=\vgcShifPlu(V_1,r',p,s')
		\]
		when $|V_1|-s'\le\finScP$, and by setting
		$T_{\finScP}[1, s', r']=0 $ when $|V_1|-s'>\finScP$ (the condition
		$|V_1|-s'\le\finScP$ is to ensure that the top-ranked candidate of
		the voters from group~$V_1$ obtains no more than~$\finScP$ points).
		We update the tables for~$i>1$ by setting $T_{\finScP}[i, s', r']$
		to be: 
		\begin{align*}
		\textstyle
		\sum_{s'' = |V_i|-\finScP}^{s'}\sum_{r'' \leq r'}
                  \big(& T_{\finScP}[i-1, s'-s'', r'-r'']
                  \cdot
		\vgcShifPlu(V_i,r'',p,s'')\big).
		\end{align*}
		(Again, the lower bound on~$s''$ ensures that the candidate the
		voters from group~$V_i$ vote for (if not~$p$) obtains no more
		than~$\finScP$ points.)  It is not hard to see that this indeed
		computes the values in the table correctly, without
		double-counting.
		
		Assuming \vgcShifPlu{} and global tables are computed correctly, it
		is not hard to verify that the overall solution is:
		$$\sum_{\finScP:=\scoreof{p}}^{\scoreof{p}+r} T_{\finScP}[m', 
		\finScP-\scoreof{p}, r].$$
		Indeed, between two different ``guesses'' of the final
		score~$\finScP$, double-counting is impossible.
	\end{proof}
	
	\begin{proof}[Proof for the destructive case]
		The main ideas behind the destructive case are very similar to those
		behind the constructive one.  The core
		observation 
		is that under destructive \textsc{Plurality \#Shift-Bribery} every
		voter will either vote for the distinguished candidate~$p$ (if $p$
		is the voter's top choice and is not shifted backward) or for some
		other candidate~$d\neq p$ (either if $p$~were not the voter's top
		choice but~$d$ were, or if $p$~were the top choice but was shifted
		backward so that~$d$, which was originally in the second position,
		was moved to the top).  In either case, we call candidate~$d$ the
		non-$p$ choice of voter~$v$.  This allows us to group all voters
		according to their non-$p$ choices as follows.  Let
		$(V_1,V_2,\ldots,V_{m'}), m'\le m$, be a partition of voters into
		groups~$V_i$ such that every two voters within the same group~$V_i$
		have the same non-$p$ choice while every two voters from different
		groups have different non-$p$ choices.
		
		Let $\scoreof{c}$ be the original score of candidate~$c$.  The idea
		of our algorithm is to count, for each possible final
		score~$\finScP \le \scoreof{p}$ of~$p$ the number of ways to spend
		the given budget~$r$ so that $p$ obtains $\finScP$ points while at
		least one other candidate obtains more than~$\finScP$ points.  For
		each possible $\finScP$ we create one global dynamic programming
		table~$T_{\finScP}$.
		
		Our global tables are defined as follows.  An entry
		$T_{\finScP}[i, s', r',\text{false}]$ contains the number of ways to
		shift~$p$ backward in the voter groups $V_1,\ldots,V_i$ at total
		cost of~$r'$, so that $p$~obtains $s'$~points from~$V_1,\ldots,V_i$
		while no non-$p$ choice of any voter group $V_1,\ldots,V_i$ receives
		more than~$\finScP$ points.  An entry
		$T_{\finScP}[i, s', r',\text{true}]$ contains the number of ways to
		shift~$p$ backward in the voter groups $V_1,\ldots,V_i$ at total
		cost of~$r'$, so that $p$~obtains $s'$~points from~$V_1,\ldots,V_i$
		while a non-$p$ choice of at least one voter group from
		$V_1,\ldots,V_i$ receives more than~$\finScP$ points.
		
		To compute the entries of the table~$T$, we will use the following
		\emph{local counting problem} maintaining the \emph{voter group
			contribution}: Given a voter group~$V'$, a budget~$r'$, a
		distinguished candidate~$p$, and a score~$s'$, compute the
		number~$\vgcDShifPlu(V',r',p,s')$ of possibilities to shift
		candidate~$p$ backward at total costs~$r'$ within~$V'$, so that
		$p$~is ranked exactly $s'$~times at the top position.  This number
		can be computed in polynomial time (see
		Lemma~\ref{lem:vgcDShiftPlu-P} in Appendix~\ref{app:vgc}).
		
		The intitialization of $T_{\finScP}$ is as follows:
		\begin{enumerate}
			\item If $|V_1|-s'\le\finScP$, then
			\[
			T_{\finScP}[1, s', r',\text{false}]=\vgcDShifPlu(V_1,r',p,s'),
			\]
			and otherwise $T_{\finScP}[1, s', r',\text{false}]=0$.
			\item If $|V_1|-s'>\finScP$ then
			\[
			T_{\finScP}[1, s', r',\text{true}]=\vgcDShifPlu(V_1,r',p,s'),
			\]
			and otherwise $T_{\finScP}[1, s', r',\text{true}]=0$. 
		\end{enumerate}
		We compute the table entries for~$i>1$ as follows.  We set
		$T_{\finScP}[i, s', r',\text{false}]$ to be
		\begin{align*}
		\textstyle
		\sum_{s'' = |V_i|-\finScP}^{s'}\sum_{r''\le r'} \big(
                  T_{\finScP}[i-1, s'-s'', r'-r'',\text{false}]
                  \cdot
		\vgcDShifPlu(V_i,r'',p,s'')\big).
		\end{align*}
		(The lower bound on~$s''$ ensures that the non-$p$ choice obtains no
		more than~$\finScP$ points.)  And we set
		$T_{\finScP}[i, s', r',\text{true}]$ to be:
		\begin{align*}
		\sum_{\substack{s''\le s', |V_i|-s''>\finScP}} \sum_{r''\le r'} (
                  T_{\finScP}[i-1, s'-s'', r'-r'',\text{false}]
		\cdot \vgcDShifPlu(V_i,r'',p,s'')) \\
		+ \sum_{\substack{s''\le s'}}\sum_{ r''\le r'} (
                  T_{\finScP}[i-1, s'-s'', r'-r'',\text{true}]
		\cdot \vgcDShifPlu(V_i,r'',p,s'')) \\
		\end{align*}
		The first two sums account for the case that the non-$p$ choice of
		group~$V_i$ is the first candidate non-$p$ candidate to obtain more
		than~\finScP{} points.  The second two sums account for the case
		where already some non-$p$ choice of some previous group obtained
		more than~\finScP{} points.  One can verify that this indeed
		computes the table correctly without double-counting.
		
		Assuming \vgcDShifPlu{} and global tables are computed correctly, it
		is not hard to verify the the overall solution is:
		$$\sum_{\finScP:=\scoreof{p}-r}^{\scoreof{p}} T_{\finScP}[m', \finScP, 
		r, \text{true}].$$
		Indeed, between two different ``guesses'' of the final score~$\finScP$,
		double-counting is impossible.
	\end{proof}
	
	\subsection{Proof of Theorem~\ref{thm:borda-shift-voters}} 
	\label{ap:borda-shift-vot}
	
	\thmbordashiftvoters*
	
	\begin{proof}
		For the constructive case, it suffices to follow the proof of
		\citet{bre-fal-nie-tal:c:shift-bribery-committees}.  For the
		destructive case, we give a Turing reduction from the
		\textsc{\#Multicolored Independent Set} problem, which is well-known
		to be $\sharpwone$-complete (indeed, \textsc{\#Independent Set} is
		equivalent to \textsc{\#Clique}, which is a canonical
		$\sharpwone$-complete problem; the multicolored variants of these
		problem remain $\sharpwone$-complete).
		
		Let $I = (G,h)$ be an instance of \textsc{\#Multicolored Independent
			Set}, where 
		$G = (V(G),E(G))$ is a graph where each vertex has one of~$h$
		colors; we ask for the number of size-$h$ independent sets (i.e.,
		sets of vertices such that no two vertices have a common edge) such
		that each vertex has a different color.
		Without loss of generality, we assume that there are no edges
		between vertices of the same color and that the number of vertices
		of each color 
		is the same, denoted by~$n$.
		For each color $\ell \in [h]$, let 
		$V(\ell) := \{v^\ell_1, \ldots, v^\ell_n\}$ denote the set of
		vertices with color~$\ell$. For each vertex $v^\ell_i \in V(G)$, let
		$E(v^\ell_i)$ denote the set of edges incident to this
		vertex.  Finally, let $\Delta := \max_{v \in V(G)}|E(v)|$ 
		be the highest degree of a vertex in~$G$.
		
		Our reduction proceeds as follows.  Let $r := h(n(\Delta+1)+\Delta)$
		be our shift radius.
		We form an election with the following candidates. First, we add a
		candidate~$d$, who will be the original winner of the election, and
		we treat sets $V(G)$ and~$E(G)$ as sets of vertex and edge candidates.
		For each vertex $v^\ell_i \in V(G)$, we form a set $F(v^\ell_i)$ of
		$\Delta - |E(v^\ell_i)|$ 
		\emph{fake-edge} candidates, so we will be able to pretend that all
		vertices have the same degree; we write $F(G)$ to denote the set of
		all fake-edge candidates.  Next, we form a \emph{blocker} candidate
		$b$ and a set $B$ of $r$ additional blocker candidates, whose
		purpose will be to limit the extent to which we can shift $d$ in
		particular votes. Finally, we let $X = \{x_1, \ldots, x_5\}$ be a
		set of five candidates that we will use
		to fine-tune the scores of the other candidates. Altogether, the
		candidate set is:
		\[ C := \{d\} \uplus V(G) \uplus E(G) \uplus F(G) \uplus B \uplus \{b\} 
		\uplus
		X. \]
		
		For each vertex $v^\ell_i$, by $H(v^\ell_i)$ we mean the
		(sub)preference order where $v^\ell_i$ is ranked on top and is
		followed by the candidates from $E(v^\ell_i) \uplus F(v^\ell_i)$ in
		some arbitrary order. We write $\overleftarrow{H(v^\ell_i)}$ to
		denote the corresponding reverse order. We form the following $4h+2$ 
		voters:
		\begin{enumerate}
			\item For each color $\ell \in [h]$, we introduce voters $e(\ell)$
			and $f(\ell)$ with preference orders:
			\begin{align*}
			\!\!\!
			e(\ell) \colon  b \pref d \pref H(v^\ell_1) \pref 
			\cdots \pref H(v^\ell_n) \pref B \pref X, \\
			f(\ell) \colon  b \pref d \pref \overleftarrow{H(v^\ell_n)} \pref 
			\cdots \pref \overleftarrow{H(v^\ell_1)} \pref B \pref X. 
			\end{align*}
			We also introduce voters $e'(\ell)$ and $f'(\ell)$, whose
			preference orders are obtained by reversing those of $e(\ell)$ 
			and~$f(\ell)$, respectively, and shifting the candidates from $X$ to
			the back (the exact order of the candidates from $X$ in the last
			five positions is irrelevant).
			
			\item 
			Let $\sigma$ be the following preference order:
			\begin{align*}
                          \sigma \colon  V(G) \pref x_1
                              \pref  E(G)  \pref F(G) \pref
                          x_2
                              \pref  x_3   \pref x_4 \pref 
			x_5 \pref B  \pref d  \pref b.
			\end{align*}
			We introduce four voters, $s_1, s_2, s_3, s_4$. Voter $s_1$ has
			preference order $\sigma$, except that members of $B$ are shifted
			ahead of $x_5$, and voter $s_2$ has the preference order obtained
			from~$\sigma$ by 
			(a)~shifting $x_1$ to the top, (b)~shifting $x_2$, $x_3$,
			and~$x_4$ ahead of the candidates from $E(G) \uplus F(G)$, and (c)
			shifting $d$ ahead of~$B$. Voters $s_3$ and $s_4$ have preference
			orders that are reverses of $\sigma$.
		\end{enumerate}
		Let $E$ be the just-constructed election.  If $s_1$ and $s_2$ had 
		preference
		order $\sigma$, then all candidates, except those in $X$, would
		have the same score (because for each voter there would be a
		matching one, with the same preference order but reversed, except
		that both voters might rank members of~$X$ on the bottom); let this
		score be~$L$. Due to the changes in $s_1$'s and $s_2$'s preference 
		orders:
		\begin{enumerate}
			\item 
			candidate $d$ has score $L+r$,
			\item 
			every vertex candidate has score $L-1$,
			\item 
			every edge and fake-edge candidate has score $L-3$,
			\item 
			every blocker candidate has score $L$, and
			\item 
			every candidate in $X$ has score much below $L-r$.
		\end{enumerate}
		
		Let~$I'$ be an instance of \textsc{Borda \#Destructive Shift-Bribery} 
		with election~$E$, designated candidate~$d$, and
		shift radius~$r$. Further, let~$M$ be the set of
		elections that can be obtained from~$E$ by shifting~$d$ back by~$r$
		positions in total, let~$f(I)$ be the number of solutions for~$I$
		(i.e., the number of multicolored independent sets of size~$h$
		in~$G$), and let~$g(I')$ be the number of solutions for~$I'$ (i.e.,
		the number of elections in $M$ where $d$ is not a winner).
		We claim that $f(I) = |M| - g(I')$.  In other words, we claim that
		each election where~$d$ wins and which can be obtained from~$E$ by
		shifting him or her back by~$r$ positions in total, corresponds to a
		unique multicolored independent set in~$G$. Since~$|M|$ can be
		computed in polynomial time using a simple dynamic program, showing
		that our claim holds will complete the proof.  First, in Step~1, we
		show that each solution for~$I$ corresponds to a unique election from 
		$M$ 
		where~$d$
		wins (and which is obtained by shifting $d$ by $r$ positions back),
		and then, in Step~2, we show that the reverse implication holds.
		
		\paragraph{Step 1.} 
		Let $S = \{v^1_{i_1}, \ldots, v^h_{i_h}\}$ be some multicolored
		independent set of $G$. We obtain a corresponding solution for~$I'$
		as follows: For each $\ell \in [h]$, we shift $d$ in $e(\ell)$ to be
		right in front of $v^\ell_{i_\ell+1}$ (or, to be right in front of
		the first blocker candidate, if $i_\ell = n$), and we shift $d$ in
		$f(\ell)$ to be right in front of $v^\ell_{i_\ell}$. Doing so
		requires $n(\Delta+1) + \Delta$ unit shifts for each $\ell \in [h]$,
		so , altogether, we make $r = h(n(\Delta+1)+\Delta)$ unit shifts.
		As a consequence, $d$ has score~$L$ and every other candidate has
		score at most~$L$. Indeed, $d$ passes each vertex candidate exactly
		once, and each edge and fake-edge candidate at most three times. The
		former is readily verifiable. The latter can be seen as follows: Fix
		a color $\ell \in [h]$ and consider voters $e(\ell)$ and $f(\ell)$.
		In their preference orders, $d$ passes each edge candidate incident
		to a vertex in $V(\ell) \setminus \{v^\ell_{i_\ell}\}$ exactly once
		(either in $e(\ell)$ or in $f(\ell)$), and $d$ passes each edge
		candidate incident to $v^\ell_{i_\ell}$ exactly twice (once in
		$e(\ell)$ and once in $f(\ell)$). Thus, each edge candidate is
		passed by $d$ at most three times (if neither of its endpoints is
		in~$S$, then it is passed twice, and if one of its endpoints is
		in~$S$, then it is passed three times; both of its endpoints cannot
		belong to~$S$ by definition of an independent set).  Similarly,
		$d$~passes each fake-edge candidate at most three times.
		Finally, $d$ never passes any of the blocker candidates. Thus, $d$~is
		a winner in the resulting election.
		
		\paragraph{Step 2.} For the other direction, consider an
		election~$E'$ obtained from $E$~by shifting~$d$ backward by
		$r$~positions in total, where $d$~still is a winner. We will show
		that $E'$ corresponds to a unique, size-$h$, multicolored
		independent set in~$G$.
		First, we recall that $d$ has score $L$ in $E'$ (this is so because
		he or she had score $L+r$ in $E$ and was shifted by $r$ positions
		backward). This means that to remain a winner, $d$ could not have
		passed any of the blocker candidates, because then some blocker
		candidate would have more than $L$
		points. 
		As a consequence, the only votes in which $d$ could have
		been shifted are $e(\ell)$ and $f(\ell)$, for each $\ell \in [h]$.
		Further, in each of these votes $d$ could have been shifted by at
		most $n(\Delta+1)$ positions. In fact, it must have been the case
		that for each $\ell \in [h]$, the total number of positions by which
		$d$ was shifted in $e(\ell)$ and $f(\ell)$ was exactly
		$r/h = n(\Delta + 1) + \Delta$. If this were not the case, then for
		some~$\ell$, candidate $d$ would have been shifted by more than 
		$n(\Delta + 
		1) + \Delta$ positions in total in $e(\ell)$ and
		$f(\ell)$ and, as a consequence, $d$ would have passed at least one
		vertex candidate from $V(\ell)$ twice. Such a vertex candidate would
		end up with score at least $L+1$ and $d$ would not have been a
		winner. To convince oneself that this is the case, let $y$ be some
		positive integer and consider vote~$e(\ell)$ with $d$ shifted by
		$n(\Delta+1)$ positions to the back (right in front of the first
		blocker candidate), and vote~$f(\ell)$ with~$d$ shifted
		by~$\Delta+y$ positions to the back. Initially $d$ passes $v^\ell_n$
		in both~$e(\ell)$ and~$f(\ell)$.  Now consider the process of
		repeatedly undoing a single unit shift in~$e(\ell)$ and performing a
		single additional unit shift in $f(\ell)$. At each point of this
		process, there is at least one vertex candidate in $V(\ell)$ such
		that $d$ is ranked behind this vertex in both
		votes. 
		
		Analogous reasoning shows that for each~$\ell \in [h]$, there is a
		number~$i_\ell \in [n]$ such that in~$e(\ell)$ candidate $d$ is
		shifted back by exactly $i_\ell(\Delta+1)$ positions, and
		in~$f(\ell)$ candidate~$d$ is shifted back by
		$(n-i_\ell)(\Delta+1) + \Delta$ positions (indeed, it suffices to
		repeat the reasoning from the end of the above paragraph for $y = 0$
		to see that these are the only numbers of shifts for which $d$ never
		passes any of the candidates from $V(\ell)$ twice).  Now we note
		that the set $S := \{ v^1_{i_1}, \ldots, v^h_{i_h}\}$ is a size-$h$,
		multicolored, independent set: The first two observations are
		immediate; for the latter, we note that---analogously to the
		reasoning in Step~1---if $S$ were not an independent set, then
		$d$ would pass some edge candidate four times, giving him or her score
		$L+1$, which would prevent $d$ from being a winner.
	\end{proof}
	
	\subsection{Proof of Theorem~\ref{thm:borda-shift-radius}}
	\thmbordashiftradius*
	
	\begin{proof}[Proof (constructive case)]
		We show how \textsc{Borda \#Constructive Shift-Bribery}
		parameterized by the radius~$r$ can be solved in $\fpt$ time using
		dynamic programming.  We start with the assumption of unit costs and
		later explain how to extend the dynamic program to work with
		arbitrary unary-encoded costs.
		
		First, observe that, given some budget~$r$ and unit costs, we know
		the final score~\finScP of candidate~$p$ after shifting it foward by
		$r$~position in total ($p$ gains one point with each position it is
		shifted forward).  Our problem becomes very easy when every other
		candidate already has score at most~\finScP (before shifting~$p$),
		because clearly no candidate other than~$p$ may gain a point.
		
		In general, since other candidates loose $r$~points in total, there
		may be up to $r$~\emph{critical candidates} that have score greater
		than~\finScP (before shifting~$p$).  Moreover, for each critical
		candidate~$c$ we can compute a \emph{demand value}~$d(c)$ which
		denotes the number of times $p$~must get shifted ahead of~$c$
		(equivalently, $d(c)$ is the original score of~$c$ minus~\finScP).
		
		Let the candidate set be $C = \{p,c_1, \ldots, c_m\}$ and, for the
		ease of presentation, assume that the candidates~$c_1, \ldots, c_m$
		are sorted by their demand values (with non-critical candidates
		having demand zero).  We define the \emph{initial demand
			vector}~$\vec{d_0}$ to be an $r$-dimensional vector of natural
		numbers where the $i$-th component~$\vec{d}[i]$ specifies how many
		times candidate~$p$ needs to pass candidate~$c_i$ to ensure that
		$c_i$ has at most score $\finScP$ (if $r$ is larger than the number
		of non-$p$ candidates, we pad the demand vector with zeros; for
		simplicity, in the further discussion we assume that $r$ is at most
		as larger as the number of non-$p$ candidates).  Thus we
		have~$\vec{d}_0=(d(c_1),\ldots,d(c_r))$.
		
		For each voter~$v$ and non-negative integer~$r'\le r$ we define the
		\emph{gain vector}
		as
		\[
		{\vec{g}(v,r')}=(g(c_1,v,r')\ldots,g(c_r,v,r')),
		\] where $g(c_i,v,r')=1$ if $p$~passes candidate~$c_i$ when $p$ is
		shifted forward at cost~$r'$ in vote~$v$.  If spending cost~$r'$ in
		vote~$v$ is impossible (e.g., because $p$ would already be pushed to
		the top position at a lower cost), then we set the gain vector
		to~$(-2r,\ldots,-2r)$.  This way, we later ensure that such
		``invalid actions'' are never counted.  Naturally, given some
		voter~$v$ and some non-negative integer~$r'\le r$, the
		vector~$\vec{d}'=\vec{d}_0-\vec{g}(v,r')$ describes the demand
		vector assuming that~$p$ was shifted forward by~$r'$ position in
		vote~$v$.  Note that there are at most~$(r+1)^{\min(r,m)}$ possible
		demand vectors.
		
		We are now ready to solve our problem via dynamic programming, using
		table~$T$ of $\fpt(r)$ size.  More precisely, let $T[i,r',\vec{d}]$
		denote the number of ways to shift~$p$ by~$r'$ positions in total,
		within the first~$i$ voters, and ending up with demand
		vector~$\vec{d}$.  The overall solution for our problem will be in
		the entry~$T[n,r,\vec{0}]$.  It remains to show how to compute the
		entries of~$T$.
		
		We do so with increasing~$i$, going over all combinations of~$r'$
		and $\vec{d}$.  Clearly, $T[1,\dots]$ has to be initialized mostly
		with zero entries since at most~$r$ different demand vectors can be
		realized.  More precisely, we have:
		$$T[1,r',\vec{d}]=\begin{cases}
		1 &\text{if } \vec{d}=\vec{d}_0 - \vec{g}(v_1,r')\\
		0 &\text{otherwise}.
		\end{cases}
		$$
		We fill-in the table for each~$i>1$ using formula:
		$$T[i,r',\vec{d}] = \sum_{\substack{r''\le r',\\\vec{d}'\le^1\vec{d}}} 
		T[i-1,r'-r'',\vec{d}'] \cdot {[\vec{d}=(\vec{d}'-g(v_i,r''))]},
		$$ where $\vec{d}'\le^1\vec{d}$ holds if vector~$\vec{d}'$ is 
		component-wise equal or smaller
		by at most one compared to~$\vec{d}$ (formally,
		$\vec{d}'\le^1\vec{d} \Leftrightarrow \forall j \in [r]: \vec{d}[j]-1
		\le \vec{d}'[j]\le\vec{d}[j]$), and where~$[[X]]$ is one if
		equation~$X$ holds and zero otherwise.  Note that this recurrence
		goes over all possible ways to distribute the budget~$r'$ among the
		first~$i-1$ voters and voter~$v_i$, while only considering demand
		vectors that can be reached with the respective budget for
		voter~$v_i$.
		
		The table size is upper-bounded by
		$n \cdot r \cdot (r+1)^{\min(r,m)}$.  Initializing a table entry
		works in $O(m)$ time (compute the gain vector and compare the demand
		vectors).  While updating the table, an entry can be computed in time
		$r \cdot 2^{\min(r,m)} \cdot O(m)$ because there are at most~$r$
		possibilities for~$r''$ and at most $2^{\min(r,m)}$~possibilities
		for~$\vec{d}'$.  Altogether, this means we can compute~$T$ and solve
		our problem in $\fpt$~time with respect to the shift radius~$r$.
		
		Finally, we explain how to extend the $\fpt$-algorithm to also work
		with arbitrary unarily encoded costs.  The crucial difference for
		non-unit costs is that we cannot compute the final score of~$p$ from
		our budget~$r$.  Instead, we guess (that is, go through all
		possibilities) the final score and then apply the algorithm described
		above with small modifications.  To ensure that~$p$ indeed ends up
		with the desired final score, we have to keep track of the score~$p$
		obtains.  This can easily be done by extending the demand (resp.\
		gain) vector by one more component that stores the number of
		times~$p$ has to pass (resp.\ passes) some candidate.
	\end{proof}
	
	\begin{proof}[Proof (destructive case)]
          The proof is provided in the main body of the paper.
        \end{proof}

	\section{Additional Material for \Cref{se:exp}}
	
	\subsection{Statistical Cultures}
	
	Here we briefly recall the four statistical cultures that we mention in
	our experimental studies:
	\begin{description}
		\item[Impartial Culture.] In the impartial culture model (IC), each 
		election consists of
		preference orders chosen uniformly at random.
		
		\item[The Urn Model.] In the urn model, with parameter $\alpha \in 
		\mathbb{R_+}$, to
		generate an election (with $m$ candidates), we start with an urn
		containing all $m!$ preference orders and generate the votes one by
		one, each time drawing the vote from the urn and then returning it
		there with $\alpha m!$ copies.
		
		\item[The Mallows Model.] In the Mallows model, with parameter $\phi 
		\in [0,1]$, each
		election has a central preference order $v^*$ (chosen uniformly at
		random) and the votes are sampled from a distribution where the
		probability of obtaining vote $v$ is proportional to
		$\phi^{d_{\text{swap}}(v^*,v)}$.
		
		\item[Euclidean Models.] In the $t$D-Cube and $t$D-Sphere models, the
		candidates and voters are points sampled uniformly at random from a
		$t$-dimensional hypercube/sphere, and the voters rank the candidates
		with respect to their distance (so a voter ranks the candidate whose
		point is closest to that of the voter first, then the next
		closest candidate, and so on).
	\end{description}
	
	Additionally, we also briefly describe the other models that are included in
	the dataset of~\citet{szu-fal-sko-sli-tal:c:map-of-elections}.
	\smallskip
	
	\noindent\textbf{Single-Peaked Elections.}\quad
	Let $C = \{c_1, \ldots, c_m\}$ be a set of candidates and let $\lhd$
	be a linear order over $C$. We will refer to $\lhd$ as the societal
	axis. We say that a voter $v$'s preference order is single-peaked with
	respect to the axis if for every $t \in [m]$ it holds that $v$'s $t$
	top-ranked candidates form an interval within $\lhd$. An election
	$E = (C,V)$ is single-peaked with respect to a given axis if each
	voter's preference order is single-peaked with respect to this axis;
	an election is single-peaked if it is single-peaked with respect to
	\emph{some} axis.
	
	The notion of single-peakedness is due to
	\citet{bla:b:polsci:committees-elections} and is intuitively
	understood as follows: The societal axis orders the candidates with
	respect to positions on some one-dimensional issue (e.g., it may be
	the level of taxation that the candidate support, or a position on the
	political left-to-right spectrum). Each voter cares only about the
	issue represented on the axis. So, each voter chooses his or her
	top-ranked candidate freely, but then the voter chooses the
	second-best one among the two candidates next to the favorite one on
	the axis, and so on.
	
	\citet{szu-fal-sko-sli-tal:c:map-of-elections} use the following two
	models for generating single-peaked elections (in both models, the
	societal axis is chosen uniformly at random and the votes are
	generated one-by-one, until a required number is produced):
	\begin{description}
		\item[Single-Peaked (Conitzer).] In the Conitzer model, we generate a
		vote as follows. First, we choose the top-ranked candidate uniformly
		at random.  Then, we perform $m-1$ iterations, extending the vote
		with one candidate in each iteration: With probability 
		$\nicefrac{1}{2}$ we
		extend the vote with the candidate ``to the left'' of the so-far
		ranked ones, and with probability $\nicefrac{1}{2}$ we extend it
		with the one ``to the right'' of the so-far ranked ones (if we ran
		out of the candidates on either side, then, naturally, we always
		choose the candidate from the other one).  This model was
		popularized by \citet{con:j:eliciting-singlepeaked} and, hence, its
		name.
		
		\item[Single-Peaked (Walsh).] In the Walsh model, we generate votes by
		choosing them uniformly at random from the set of all preference
		orders single-peaked with respect to a given axis. This model was
		popularized by \citet{wal:t:generate-sp}, who also provided a
		sampling algorithm.
	\end{description}
	\citet{szu-fal-sko-sli-tal:c:map-of-elections} give a detailed
	analysis explaining why these two models produce quite different
	elections (they also point out that 1D-Interval elections tend to be
	very similar to single-peaked elections from the Conitzer model).
	\smallskip
	
	\noindent\textbf{Elecitons Single-Peaked on a Circle (SPOC).}\quad
	\citet{pet-lac:j:spoc} extended the notion of single-peaked elections
	to single-peakedness on a circle. The model is very similar to the
	classic notion of single-peakedness, except that the axis is cyclic.
	Let $C$ be a set of $m$ candidates.  Voter $v$ has a preference order
	that is single-peaked on a circle with respect to the axis $\lhd$ if
	for every $t \in [m]$ it holds that the set of $t$ top-ranked
	candidates according to $v$ either forms an interval with respect to
	$\lhd$ or a complement of an interval.
	
	\citet{szu-fal-sko-sli-tal:c:map-of-elections} generate SPOC elections
	in the same way as single-peaked elections in Conitzer's model, except
	that the axis is cyclic (so one never ``runs out of candidates'' on
	one side). Such SPOC elections are quite similar to 2D-Hypersphere
	ones (and, as indicated by Szufa et al., also to IC
	elections).\smallskip
	
	\noindent\textbf{Single-Crossing Elections.}\quad
	Intuitively, an election is single-crossing if it is possible to order
	the voters so that for each two candidates $a$ and $b$, as we consider
	the voters in this order, the relative ranking of $a$ and $b$ changes
	at most once. Formally, single-crossing elections are defined as
	follows.
	
	Let $E = (C,V)$ be an election, where $C = \{c_1, \ldots, c_m\}$ and
	$V = (v_1, \ldots, v_n)$. This election is single-crossing with
	respect to its natural voter order if for each two candidates $c_i$,
	$c_j$ there is an integer $t_{i,j}$ such that the set
	$\{ \ell \mid v_\ell$ ranks $c_i$ above $c_j\}$ is either
	$\{1,2, \ldots, t_{i,j}\}$ or $\{t_{i,j}, \ldots, n-1, n\}$.  An
	election is single-crossing if it is possible to reorder its voters so
	that it becomes single-crossing with respect to the natural voter
	order.
	
	It is not clear how to generate single-crossing elections uniformly at
	random, or what a good procedure for generating single-crossing
	elections should be.  \citet{szu-fal-sko-sli-tal:c:map-of-elections}
	propose one procedure and we point the reader to their paper for the
	details.
	
	\subsection{Sampling Elections}\label{app:samp}
	In our experiments, to calculate $P_{E',c}(r)$ for different
	candidates $c$, elections $E'$, and swap distances $r$, we sampled
	elections at swap distance $r$ from $E'$ uniformly at
	random. Unfortunately, to achieve this, it is not enough to simply
	perform $r$ swaps in some of the votes in $E'$, as this procedure does
	not necessarily produce an election at distance $r$ and cannot be
	easily adapted to result in a uniform distribution. Thus, we use a
	sampling procedure that relies on counting the number of elections at
	some swap distance $r$.
	
	To compute this value, we employ dynamic programming using a table
	$T^m_E$: Each entry $T^m_E[n,r]$ contains the number of elections at
	swap distance $r$ from a given fixed election with $n$ voters and $m$
	candidates (note that it is irrelevant what this election is, so we
	can simply assume an election with $n$ identical votes). In the
	following, let $u(m)=\frac{m(m-1)}{2}$ denote the maximal number of
	swaps that can be performed in a vote.
	
	To compute $T^m_E$, we start by computing $T^m_E[1,r]$, that is, the
	number of votes at swap distance $r$ from a given single vote over $m$
	candidates. We compute this value using dynamic programming. As we
	will use this quantity separately in the following sampling algorithm,
	we create a separate table $T_V$ for it, where $T_V[m,r]$ contains the
	number of votes at swap distance $r$ from a given vote over $m$
	candidates. Note that $T_V[m,r]$ is simply the number of permutations
	over $m$ elements with $r$ inversions. Computing this value is a
	well-studied problem and we use the following procedure \citep{oeis}:
	We initialize the table with $T_V[m,0]=1$. We update the table by
	increasing $m$ and for each $m$ starting from $r=0$ and going to
	$r=u(m)$ using the following recursive relation:
	\[
	T_V[m,r]=T_V[m,r-1]+T_V[m-1,r]-T_V[m-1,r-m].
	\]
	Unfortunately, no closed form expression for this value seems to be
	known \citep{oeis}.
	
	Using $T_V$, we are now ready to compute $T^m_E$. We start by setting 
	$T^m_E[1,r]=T_V[m,r]$. Subsequently, we fill $T^m_E$ by increasing $n$ and 
	for 
	each $n$ starting from $r=0$ and going to $r=u(m)\cdot n$ using the 
	following 
	recursive relation:
	$$T^m_E[n,r]=\smashoperator{\sum_{i=\max(r-u(m)\cdot(n-1),0)}^{\min(r,u(m))}
	} 
	T_V[m,i]\cdot
	T^m_E[n-1,r-i].$$
	The reasoning behind this formula is that, in the $n$th vote, 
	between $\max(r-u(m)\cdot(n-1),0)$ and $\min(r,u(m))$ swaps can be 
	performed. 
	We iterate over all these possibilities and count, for each $i$ in this 
	interval, the number of elections where $i$ swaps in the $n$th vote and 
	$r-i$ 
	swaps in the remaining $n-1$ votes are performed.
	
	Using $T_E$ and $T_V$, we split the process of sampling elections at
	swap distance $r$ into two steps. First, we sample the distribution of
	swaps to votes proceeding recursively vote by vote. For the first
	vote, the probability that $i$ swaps are performed is proportional to
	the number of possibilities to perform $i$ swaps in the first vote times the
	number of elections at swap distance $r-i$ from the remaining
	$n-1$-voter election. This results in performing
	$i\in [\max(r-u(m)\cdot(n-1),0),\min(r,u(m))]$ swaps in the first vote
	with probability
	$$\frac{T_V[m,i]\cdot T^m_E[n-1,r-i]}{T^m_E[n,r]}.$$ We then delete this 
	vote 
	and solve the problem for 
	the remaining $n-1$ votes and $r-i$ swaps recursively. 
	
	Second, for each vote, we generate a vote at the assigned swap distance 
	$r$. To 
	do so, we generate a 
	permutation with $r$ inversions uniformly at random and sort the vote 
	according 
	to the permutation. A permutation $\sigma$ over $m$ elements is fully 
	characterized by the tuple $t(\sigma)=(t_1(\sigma),\dots,t_m(\sigma))$ 
	where, 
	for each 
	$i\in [m]$, $t_i(\sigma)$ denotes the number of entries smaller than 
	$i$ appearing 
	after $i$ in $\sigma$. Note that, for all $i\in [m]$, it needs to hold that 
	$0\leq t_i(\sigma)\leq m-i$ and that the number of inversions in a 
	permutation 
	$\sigma$ corresponds to the sum of the entries of $t(\sigma)$. Thus, the 
	problem 
	of uniformly sampling a permutation over $m$ elements with $r$ inversions 
	is 
	equivalent to uniformly distributing $r$ indistinguishable balls into $m-1$ 
	distinguishable bins with each bin $i\in [m-1]$ having capacity $i$. We 
	again proceed bin by bin: We put 
	$i\in [0,m-1]$ balls into the 
	first bin with probability $\frac{T_V[m-1,r-i]}{T_V[m,r]}$ and then solve 
	the 
	problem for the remaining $m-2$ bins and $r-i$ balls recursively. 
	
	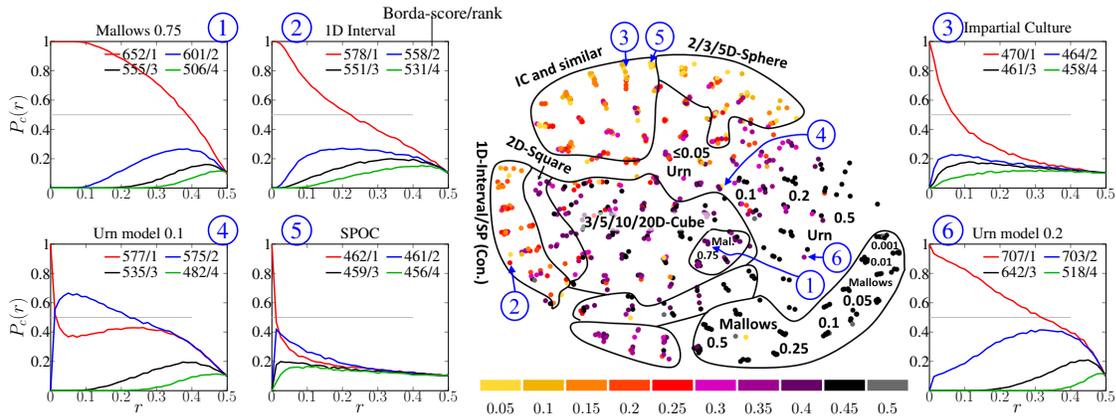
\begin{figure*}[t]
		\begin{center}
			\resizebox{\textwidth}{!}{\input{l_image_borda}}
		\end{center}
		\caption{Map of elections visualizing the $50\%$-winner threshold
			for Borda and six plots showing $P_{E,c}(r)$ as a
			functon of $r$, for six selected elections and the four most
			successful candidates in each (see the paragraph preceding
			\Cref{con:winners-winning} in the main body for
			details).}\label{fig:mapbord}
	\end{figure*}
	
	\subsection{Additional Details for Experiments on Borda} \label{app:borda}
	
	Let us now discuss the results regarding Borda elections in more
	detail. In \Cref{fig:mapbord}, we display the results of our
	experiments for Borda on elections with $10$ candidates and $100$
	voters. As in \Cref{fig:mapplur}, the ``map of elections'' shows the
	$50\%$-winner threshold and the six plots show the probability of
	victory of four candidates in six selected elections (the
	interpretation of these plots is the same as in the case of Plurality
	and \Cref{fig:mapplur}).
	
	Comparing the map for Plurality (from \Cref{fig:mapplur}) and Borda
	(from \Cref{fig:mapbord}), we see that \Cref{f1} still holds: Borda
	winners tend to be more robust than winners under Plurality, as
	witnessed by the fact that for most elections the $50\%$-winner
	threshold for Borda is higher than for Plurality. This is particularly
	true for $t$D-Cube elections and Urn elections with high values of the
	contagion parameter. Moreover, as in \Cref{fig:mapplur}, there is a
	correlation between the position of an election on the map and the
	robustness of its winners (\Cref{f3}). Most of the other patterns
	described for Plurality can also be found in the map for Borda; some
	of them are even more pronounced, as, e.g., the difference between
	1D-Interval and multidimensional hypercube elections.
	
	Considering the robustness of election winners, it turns 
	out that the score difference between the winner and its runner-up
	has a higher predictive value for Borda than for 
	Plurality.   
	\begin{conclusion}
		Candidates who win by more than $50$ points are very robust
		winners. The winning probability of candidates who win by around
		$20$ points might decrease noticeably already at small swap
		distance. Nevertheless, these candidates typically stay the most
		probable winner for a long time.
	\end{conclusion}
	Let us consider the six elections from \Cref{fig:mapbord}.  Election~1
	has the largest score difference between the winner and the runner-up
	where at swap distance $0.1$ the \second/ candidate has
	some non-negligible probability of victory. Nevertheless, the \first/
	candidate is the clear winner of this elections, even after many
	random swaps. Note that this finding may be intuitively surprising, as
	$50$ points seems relatively little as compared to the totally awarded
	$4500$ points (yet, in a Borda election with $10$ candidates and $100$
	voters, the highest possible score is $900$ and compared to this
	value, $50$ is not completely negligible).
	
	In contrast to this, in Election 2, where the score difference between
	the winner and the runner-up is around $20$, the \second/ candidate
	starts to have a non-negligible chance of winning even at swap
	distance $0.025$. Nevertheless, the \first/ candidate stays the most
	probable winner until swap distance $0.5$ and wins with more than 50\%
	probability until swap distance $0.2$. Thus, there is little doubt
	that the \first/ candidate should be the winner of the election.
	
	\begin{conclusion}
		Candidates winning by less than $10$ points might be both quite
		robust or very sensitive to random swaps.
	\end{conclusion}
	As soon as the score difference between the first and second candidate
	drops to around $10$ points, elections with a similar distribution of
	scores start to show fundamentally different behavior. Only examining
	the scores, Elections 3 to 6 all seem close, as the maximum gap
	between the \first/ and \second/ candidates is at most six points in
	these elections. However, looking at the plots, there are significant
	differences. For Elections 4 and 5, the initial winner stops to be the
	most probable winner already at swap distance $0.0125$, which is the
	smallest swap distance we examined. In Election 5, all candidates lie
	close to each other. In contrast to this, in Election 4, the \second/
	candidate dominates the \first/ candidate at swap distance $0.0125$
	and above to an extent that one might wonder whether the \second/
	candidate is not the ``true'' winner of the election (even though the
	election was generated from the urn model and was not modified in any
	way).
	
	Elections~4 and~5 stand in sharp contrast to the other two ``close''
	elections, i.e., Elections~3 and~6. For example, Elections~3 and~5,
	which both have a similar distribution of scores and lie close on the
	map, exhibit quite a different behavior. In Election~3, the \first/
	candidate remains the most probable winner until swap distance $0.4$
	and in Election~5 the \first/ candidate stops being the most probably
	winner very early on. In Election~6, the \first/ candidate stays a
	winner with probability greater than $50\%$ even until swap distance
	$0.3$. To sum up, despite all seeming quite close, while for Elections
	4 and 5 it is recommendable to reexamine the election issue, for
	Election 3 and, in particular, for Election 6, the selected \first/
	candidate is quite a robust winner.
	
	\begin{figure*}
		\centering
		\begin{subfigure}[b]{0.48\textwidth}
			\centering            
			\resizebox{6cm}{!}{\input{50w_full_plur.tex}}
			\caption{Plurality}
			\label{fig:dist50wplur}
		\end{subfigure}
		\hfill
		\begin{subfigure}[b]{0.48\textwidth}
			\centering \resizebox{6cm}{!}{\input{50w_full_bord.tex}}
			\caption{Borda}
			\label{fig:dist50wborda}
		\end{subfigure}
		\caption{Number of elections
			(on the $y$-axis) with a given $50\%$-winner threshold (on the
			$x$-axis) on the five considered datasets. The number of voters is 
			denoted by $n$ and the number of candidates by $m$}
		\label{fig:dist50w}
	\end{figure*}
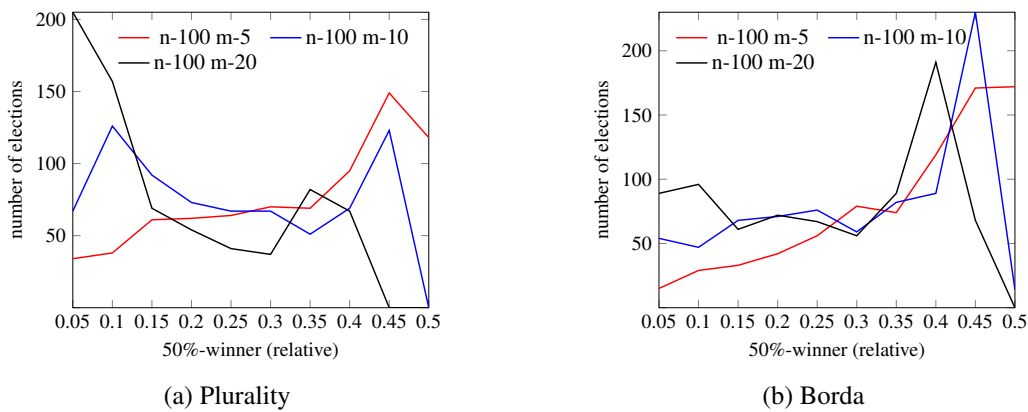
	
	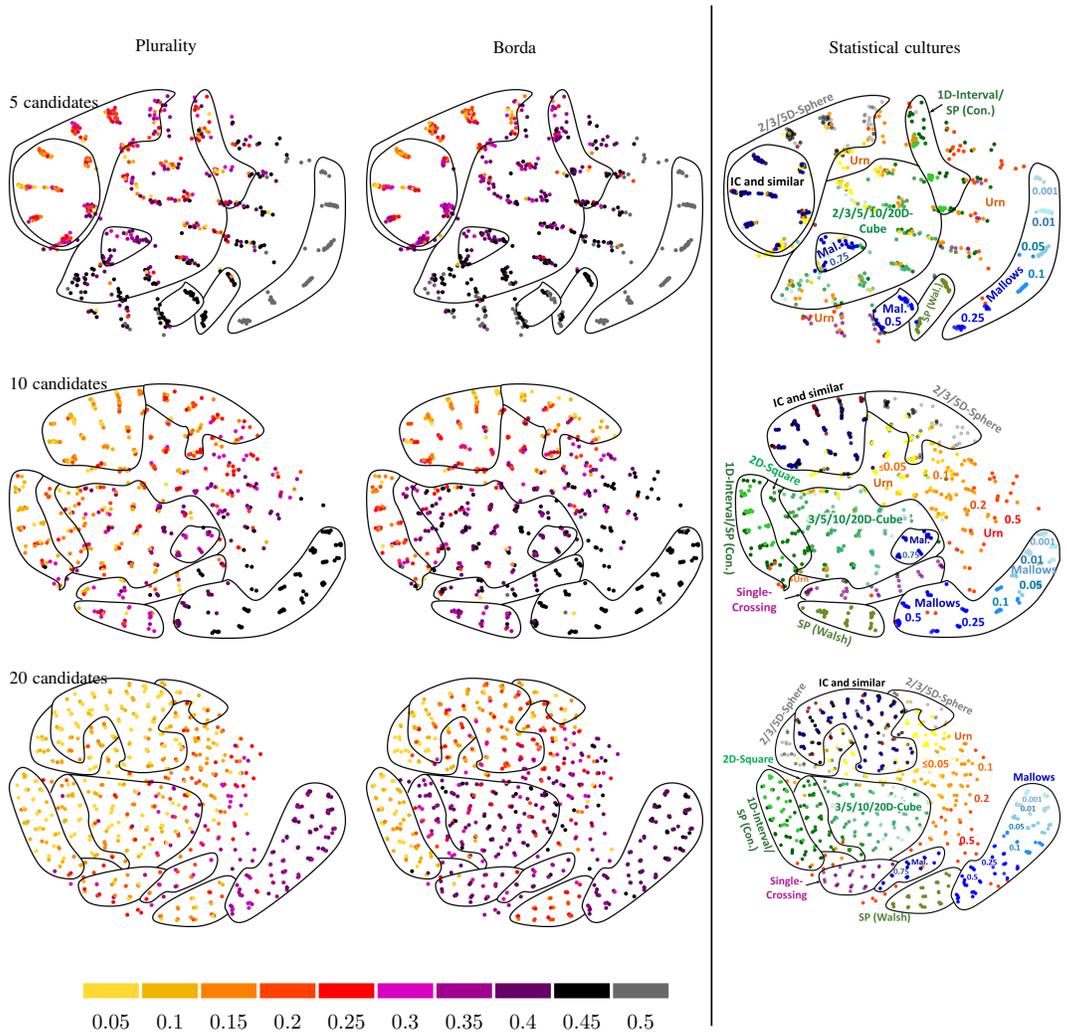
\begin{figure*}[htp]
		\begin{center}
			\resizebox{\textwidth}{!}{\input{l_image_append}}
		\end{center}
		\caption{Maps showing the $50\%$-winner thresholds for
			Plurality and Borda, for the datasets with $5$, $10$, and
			$20$ candidates and $100$ voters.}\label{fig:mapapp}
	\end{figure*}  
	
	\subsection{Experiments on Other Datasets} \label{app:diff_data}
	Besides the experiments on the $800$-elections dataset with $100$
	voters and $10$ candidates from
	\citet{szu-fal-sko-sli-tal:c:map-of-elections} that we presented in
	the main body (and in the preceding section), we also conducted
	experiments on the datasets with $100$-voters and either $5$ or $20$
	candidates, also due to
	\citet{szu-fal-sko-sli-tal:c:map-of-elections}.  In the following, we
	briefly compare the results for the different datasets with respect to
	the $50\%$-winner threshold and analyze the influence of the number of
	candidates on it.  In \Cref{fig:dist50w}, we display the number of
	elections with a particular 50\%-winner threshold from the different
	datasets. In \Cref{fig:mapapp}, we show the map of elections
	visualizing the $50\%$-winner threshold for Plurality and Borda.
	
	\begin{conclusion}
		For Plurality, increasing the number of candidates makes winners
		significantly less robust. For Borda, the same effect is visible, but 
		is far 
		less 
		pronounced.
	\end{conclusion}
	Looking at the distribution of the $50\%$-winner threshold for
	Plurality, depicted in \Cref{fig:dist50wplur}, winners in the
	$5$-candidates dataset (red solid line) are the most robust
	ones. Increasing the number of candidates to $10$ (blue solid line) or
	$20$ (black solid line), the winners become less and less robust. In
	fact, for the case of $20$-candidates, winners in a majority of
	elections have a $50\%$-winner threshold of at most $0.15$. The same
	effect can also be seen when looking at the maps for Plurality from
	\Cref{fig:mapapp}. Here, as we examine the left column from top to
	bottom, we see that the dots become lighter and lighter.
	
	One possible explanation why the number of candidates has such a
	strong effect on the $50\%$-winner threshold for Plurality is that,
	for most models, the Plurality score of the winner is
	inverse-proportional to the number of candidates. For instance, with
	$5$ candidates and $100$ voters, the average candidate score is $20$,
	whereas for $20$ candidates and $100$ voters, the average score is
	around $5$. Thus, in elections with more candidates, the absolute
	differences between the scores of the winners and the other candidates
	are often much smaller (albeit this is not true for all
	distributions).  In effect, usually fewer swaps suffice to exchange
	the winner. As we consider normalized quantities, this effect is even
	intensified by the quadratic increase of the total number of swaps
	with respect to the number of candidates.  For Borda, the same effect (but 
	less 
	strong)
	is present.
	
	\bigskip
	Examining the maps from \Cref{fig:mapapp}, the patterns described in the 
	main 
	body for Plurality 
	are also present in the other maps. In particular, on all displayed
	maps, moving from the top to the bottom and from left to right,
	elections start to have more and more robust winners. However, the
	visibility of other patterns, especially of the similarity of
	elections from the same model, varies. For instance, for Plurality on
	the $20$-candidates dataset (left bottom) most elections have a
	non-robust winner, which makes it hard to distinguish between
	elections from different models. For $5$~candidates (top), on the
	other hand, the embedding does not separate elections from different
	models nicely, which makes it hard to draw any fine-grained
	conclusions. In contrast to this, for Borda and the $20$-candidates
	dataset, the separation of the different models is clearly visible and
	there are several patterns that can be found here that are only weakly
	visible on the other maps.  For instance, Walsh single-peaked
	elections have far less robust winners than elections from other
	models lying close to them on the map.

\end{document}

%% file: scatterplur.tex
\begin{tikzpicture}

\begin{axis}[
axis y line*=right, label style={font=\Large},
xlabel=50\%-winner threshold,
ylabel=score difference,
x grid style={white!69.0196078431373!black},
xmin=0, xmax=0.52,
xtick style={color=black},
xtick={0,0.05,0.1,0.15,0.2,0.25,0.3,0.35,0.4,0.45,0.5},
xticklabels={,0.05,0.1,0.15,0.2,0.25,0.3,0.35,0.4,0.45,0.5},
y grid style={white!69.0196078431373!black},
ymin=-3.95, ymax=104.95,
ytick style={color=black},ylabel shift = -13pt
]
\addplot [only marks, mark=*, draw=black, fill=black, 
colormap/viridis,mark size=0.9pt]
table{%
x                      y
0.05 2
0.2 7
0.1 2
0.15 3
0.25 5
0.1 1
0.1 5
0.0992 2
0.0992 1
0.1492 3
0.1008 2
0.05 1
0.25 7
0.15 2
0.1008 1
0.1 3
0.0492 1
0.2 5
0.0984 1
0.1016 1
0.0976 1
0.1024 1
0.2 6
0.0508 1
0.0484 1
0.15 1
0.0516 1
0.0968 1
0.15 5
0.3 8
0.25 10
0.1032 1
0.1992 5
0.0984 2
0.2492 5
0.3 1
0.1492 2
0.1508 3
0.15 7
0.25 6
0.1 4
0.2 3
0.0476 1
0.0524 1
0.2008 5
0.2 10
0.4 24
0.15 4
0.1992 7
0.2 4
0.15 10
0.0992 4
0.2 8
0.15 6
0.096 1
0.3 9
0.104 1
0.1992 10
0.2 2
0.15 8
0.35 15
0.3 18
0.25 17
0.1492 1
0.1508 1
0.2 24
0.0992 3
0.25 12
0.25 4
0.2492 6
0.1492 4
0.0492 2
0.1492 7
0.0508 2
0.1992 3
0.2 12
0.2 16
0.1 10
0.1484 1
0.1008 3
0.1508 2
0.35 6
0.1492 5
0.35 3
0.1992 16
0.2008 3
0.0468 1
0.1016 2
0.1516 1
0.4 1
0.05 3
0.3 6
0.0992 5
0.2 1
0.4 8
0.35 13
0.1008 4
0.25 8
0.0484 2
0.0532 1
0.2 11
0.046 1
0.1 7
0.2008 10
0.4 11
0.15 11
0.0952 1
0.0976 2
0.1048 1
0.3492 13
0.2 9
0.3992 24
0.1508 5
0.3 5
0.2508 6
0.0984 4
0.1984 3
0.25 3
0.1476 1
0.1508 4
0.3492 6
0.3 7
0.35 9
0.1992 1
0.1992 6
0.45 10
0.0944 1
0.3 10
0.4 2
0.0984 3
0.0492 3
0.1492 6
0.2992 6
0.3 14
0.3492 3
0.4 9
0.1056 1
0.2492 12
0.0508 3
0.1992 4
0.4 3
0.3 2
0.3992 9
0.3992 8
0.35 4
0.2992 9
0.054 1
0.4 15
0.4 10
0.3992 1
0.1484 5
0.3 3
0.3992 10
0.1024 2
0.0936 1
0.0968 2
0.3992 2
0.2992 1
0.4 12
0.3508 3
0.4 6
0.35 5
0.35 8
0.3008 6
0.1984 10
0.0452 1
0.1064 1
0.2008 1
0.4 17
0.4 23
0.0516 2
0.2492 4
0.3992 11
0.1508 6
0.0928 1
0.3492 8
0.0548 1
0.3008 9
0.45 38
0.3992 12
0.35 7
0.4 4
0.1072 1
0.1 6
0.35 1
0.2484 6
0.3992 6
0.45 22
0.3508 13
0.4008 12
0.4008 6
0.1984 5
0.2492 8
0.1032 2
0.4 5
0.1508 7
0.1992 2
0.3 13
0.1484 7
0.1992 9
0.1524 1
0.25 14
0.1484 2
0.0444 1
0.2008 6
0.2992 13
0.2992 7
0.1516 2
0.2492 7
0.1476 2
0.092 1
0.2008 4
0.108 1
0.1984 4
0.2016 5
0.1484 4
0.0556 1
0.25 9
0.2508 12
0.2492 3
0.0912 1
0.1484 6
0.1088 1
0.3 4
0.3 19
0.1524 2
0.0436 1
0.096 2
0.1516 5
0.1476 5
0.2008 2
0.3008 7
0.1468 2
0.1524 5
0.2492 9
0.0564 1
0.1516 6
0.104 2
0.0904 1
0.25 1
0.0952 2
0.1984 1
0.25 13
0.1532 2
0.2008 7
0.146 2
0.1984 2
0.1048 2
0.0944 2
0.1984 6
0.2016 6
0.2508 8
0.2016 4
0.0428 1
0.1016 3
0.1056 2
0.2016 1
0.1516 4
0.1468 1
0.1476 4
0.0936 2
0.2016 3
0.2484 8
0.1096 1
0.1064 2
0.0928 2
0.1016 4
0.2492 13
0.1976 6
0.154 2
0.1484 3
0.0896 1
0.0572 1
0.042 1
0.1468 5
0.1516 7
0.1516 3
0.1476 3
0.2508 5
0.1072 2
0.0976 3
0.092 2
0.1532 5
0.0976 4
0.146 5
0.058 1
0.1452 2
0.1024 3
0.154 5
0.0476 2
0.1104 1
0.1548 2
0.1976 5
0.1524 4
0.0524 2
0.0968 3
0.1024 4
0.1468 4
0.0412 1
0.4008 9
0.2 18
0.1492 8
0.4 14
0.3992 3
0.2984 7
0.45 34
0.2492 17
0.45 50
0.2484 5
0.0992 7
0.45 19
0.2992 18
0.0588 1
0.0404 1
0.0992 10
0.35 23
0.35 16
0.108 2
0.2 19
0.4008 11
0.0468 2
0.1008 7
0.1984 7
0.3008 18
0.1008 5
0.2508 9
0.4 31
0.2992 5
0.1032 3
0.096 3
0.0968 4
0.1032 4
0.0596 1
0.0912 2
0.096 4
0.0888 1
0.1088 2
0.1112 1
0.0396 1
0.0604 1
0.2508 3
0.104 3
0.0388 1
0.0904 2
0.1096 2
0.0612 1
0.038 1
0.088 1
0.1976 4
0.112 1
0.0872 1
0.2992 2
0.2992 4
0.3008 5
0.3016 7
0.3008 2
0.2976 7
0.1444 2
0.3 11
0.0484 3
0.1476 7
0.3024 7
0.2508 4
0.3008 4
0.2992 10
0.3008 10
0.0952 3
0.35 10
0.35 12
0.2984 6
0.25 2
0.3492 16
0.2484 3
0.3492 10
0.1128 1
0.1976 1
0.2984 4
0.1556 2
0.3484 3
0.35 22
0.4 65
0.45 81
0.4008 24
0.3 23
0.2484 4
0.3 21
0.45 56
0.0516 3
0.4 45
0.2024 4
0.2 22
0.45 28
0.2484 12
0.3008 13
0.45 49
0.45 80
0.4 25
0.1532 4
0.1 12
0.25 16
0.35 25
0.4 52
0.1532 1
0.35 50
0.3 36
0.45 88
0.1976 3
0.45 17
0.3 17
0.1 8
0.0896 2
0.3984 12
0.45 29
0.2508 17
0.45 39
0.0476 3
0.45 61
0.4 46
0.062 1
0.3984 24
0.3992 31
0.2492 1
0.2984 5
0.35 18
0.1992 19
0.3992 14
0.35 32
0.25 22
0.45 30
0.0992 6
0.1104 2
0.1452 5
0.4 39
0.4008 14
0.4492 56
0.3992 15
0.4492 34
0.3508 8
0.146 1
0.2484 17
0.3016 4
0.2992 3
0.1992 8
0.3016 6
0.2516 5
0.3484 13
0.2024 1
0.3008 3
0.3 28
0.2508 13
0.3516 13
0.25 11
0.1008 6
0.3984 6
0.4 41
0.3476 13
0.2992 17
0.35 24
0.2992 14
0.4 26
0.3492 25
0.2016 7
0.35 17
0.146 4
0.35 11
0.2008 8
0.3508 25
0.3508 6
0.0532 2
0.1048 3
0.2976 4
0.4 42
0.3008 17
0.3 15
0.25 19
0.1436 2
0.3492 18
0.2492 11
0.3492 15
0.2508 1
0.2968 7
0.0888 2
0.2484 1
0.2516 1
0.046 2
0.3016 5
0.4008 15
0.0984 5
0.2516 12
0.2024 6
0.4008 3
0.4 21
0.3508 15
0.2476 5
0.054 2
0.2516 4
0.2484 9
0.154 1
0.35 14
0.0372 1
0.2492 10
0.0944 3
0.1524 3
0.2492 14
0.154 4
0.2516 6
0.0628 1
0.1112 2
0.0364 1
0.0452 2
0.35 2
0.2008 9
0.3032 7
0.088 2
0.2476 1
0.2016 2
0.1468 3
0.1564 2
0.1976 2
0.1984 9
0.1452 1
0.104 4
0.1428 2
0.1056 3
0.0636 1
0.112 2
0.1524 7
0.0356 1
0.2516 3
0.1968 4
0.1548 5
0.2024 5
0.0872 2
0.2024 2
0.1968 5
0.2976 5
0.0864 1
0.1548 1
0.1128 2
0.0548 2
0.0864 2
0.1136 1
0.0936 3
0.0952 4
0.2492 2
0.1444 5
0.1968 1
0.0644 1
0.0856 1
0.2032 4
0.1532 3
0.1136 2
0.2032 5
0.2476 6
0.0348 1
0.0444 2
0.1144 1
0.0848 1
0.1452 4
0.1048 4
0.0856 2
0.0652 1
0.196 5
0.1016 5
0.1152 1
0.034 1
0.1548 4
0.1444 4
0.084 1
0.066 1
0.0332 1
0.0668 1
0.116 1
0.1556 4
0.146 3
0.0832 1
0.2508 7
0.1508 8
0.1144 2
0.2524 1
0.0324 1
0.2524 6
0.2476 3
0.1444 1
0.1168 1
0.2024 3
0.1556 1
0.3024 4
0.1436 1
0.0676 1
0.0316 1
0.0824 1
0.0556 2
0.1564 1
0.1176 1
0.1436 4
0.2484 7
0.3492 14
0.3492 4
0.3492 9
0.3508 4
0.2468 6
0.3992 4
0.3492 12
0.3024 5
0.2992 11
0.3508 9
0.296 7
0.2984 3
0.2984 2
0.3492 2
0.2976 6
0.3984 3
0.3024 6
0.3484 9
0.4 29
0.4 13
0.4 35
0.3992 23
0.45 37
0.4 22
0.4 32
0.3984 15
0.4 16
0.4 30
0.3992 21
0.4008 23
0.3992 26
0.3984 23
0.4492 39
0.45 36
0.4 19
0.45 21
0.4 28
0.3992 30
0.45 55
0.45 65
0.45 52
0.45 64
0.45 69
0.45 54
0.45 57
0.45 60
0.45 43
0.45 44
0.4492 65
0.4492 54
0.4492 60
0.4508 65
0.4492 55
0.45 59
0.4508 56
0.4508 55
0.4492 52
0.45 53
0.45 79
0.4492 81
0.45 72
0.45 75
0.45 73
0.45 85
0.4492 79
0.4492 72
0.4508 81
0.45 84
0.45 83
0.4484 81
0.45 92
0.45 82
0.4492 84
0.4508 79
0.4492 83
0.4508 72
0.4492 82
0.4508 83
0.45 94
0.4492 88
0.45 96
0.4492 80
0.4508 88
0.4484 88
0.4492 92
0.4516 88
0.45 89
0.4508 84
0.4484 84
0.4476 88
0.4508 92
0.45 87
0.4492 89
0.4492 96
0.45 95
0.4484 92
0.4508 96
0.4508 89
0.45 100
0.4492 100
0.4508 100
0.45 98
0.4484 96
0.4492 98
0.4508 98
0.4484 100
0.4516 100
0.4476 100
0.4524 100
0.4484 98
0.4516 98
0.4468 100
0.4532 100
0.4516 92
0.4492 94
0.4476 98
0.446 100
0.454 100
0.4452 100
0.4524 98
0.4548 100
0.4444 100
0.4556 100
0.4436 100
0.4564 100
0.4428 100
0.4468 98
0.4532 98
0.4572 100
0.442 100
0.458 100
0.4412 100
0.4588 100
0.4404 100
0.4596 100
0.4396 100
0.4604 100
0.4388 100
};
\end{axis}
\begin{axis}[
color=blue,  label style={font=\Large},
separate axis lines,
  axis y line*=left,
  axis x line=none,
  ylabel=number of elections,
  ymin=-9, ymax=260,
  xmin=0, xmax=0.52,ylabel shift = -5pt,every axis 
plot/.append style={thick}
]
\addplot coordinates{
	(0.05, 67) (0.1, 126) (0.15, 92) 
(0.2, 73) (0.25, 67) (0.3, 67) (0.35, 51)
(0.4, 69) (0.45, 123) (0.5, 0)
};
\end{axis}

\end{tikzpicture}

%% file: scatterborda.tex
\begin{tikzpicture}

\begin{axis}[
axis y line*=right, label style={font=\Large},
xlabel=50\%-winner threshold,
ylabel=score difference,
x grid style={white!69.0196078431373!black},
xmin=0, xmax=0.52,
xtick style={color=black},
xtick={0,0.05,0.1,0.15,0.2,0.25,0.3,0.35,0.4,0.45,0.5},
xticklabels={,0.05,0.1,0.15,0.2,0.25,0.3,0.35,0.4,0.45,0.5},
y grid style={white!69.0196078431373!black},
ymin=-3.95, ymax=265,
ytick style={color=black},ylabel shift = -8pt
]
\addplot [only marks, mark=*, draw=black, fill=black, 
colormap/viridis,mark size=0.9pt]
table{%
x                      y
0.1 5
0.2 21
0.15 7
0.25 12
0.35 20
0.15 3
0.25 30
0.3 49
0.3 13
0.1 7
0.2 9
0.15 15
0.25 31
0.2 11
0.15 10
0.25 6
0.3 56
0.3 38
0.2 19
0.25 33
0.1 6
0.1 8
0.1 3
0.0992 7
0.25 18
0.3 37
0.15 9
0.0992 8
0.1492 10
0.1992 11
0.4 25
0.05 1
0.1492 7
0.1992 9
0.15 2
0.3 15
0.2008 9
0.2 3
0.1508 7
0.25 11
0.25 20
0.35 32
0.05 4
0.2492 11
0.25 16
0.25 19
0.2 5
0.1 4
0.45 83
0.2508 11
0.35 29
0.15 6
0.15 5
0.3 24
0.3 12
0.0992 4
0.2 8
0.1492 5
0.25 9
0.2 10
0.4 42
0.35 11
0.35 19
0.35 9
0.0492 1
0.3 21
0.4 64
0.3 17
0.3 19
0.4 47
0.35 30
0.2484 11
0.4 18
0.3 10
0.35 41
0.4492 83
0.35 25
0.4 19
0.2 4
0.35 37
0.4 50
0.45 123
0.4 23
0.15 1
0.4 21
0.4 32
0.45 76
0.1484 7
0.45 74
0.1992 10
0.35 40
0.25 2
0.45 126
0.45 80
0.35 44
0.4 76
0.4 96
0.45 110
0.4492 74
0.45 146
0.2992 24
0.45 135
0.4 49
0.3992 23
0.4 39
0.4 34
0.45 55
0.35 22
0.25 17
0.3492 44
0.45 160
0.3492 41
0.4 53
0.3992 25
0.4492 80
0.4 40
0.35 39
0.3 9
0.3 29
0.2492 17
0.45 145
0.45 66
0.35 8
0.4 33
0.4 68
0.35 13
0.3 14
0.35 26
0.45 101
0.45 73
0.2 13
0.4 83
0.45 173
0.4 60
0.25 5
0.45 162
0.4 48
0.4 38
0.45 93
0.4 55
0.3992 39
0.4492 76
0.3492 29
0.1492 2
0.45 94
0.45 91
0.3992 55
0.45 58
0.45 138
0.45 125
0.4 56
0.45 119
0.4 88
0.4508 83
0.4492 119
0.45 113
0.3 52
0.35 50
0.2992 17
0.4492 110
0.45 79
0.45 61
0.35 27
0.35 21
0.45 102
0.3492 32
0.4 59
0.35 3
0.25 23
0.45 69
0.4508 74
0.45 197
0.3992 47
0.45 38
0.3992 18
0.45 33
0.35 18
0.1492 9
0.2992 14
0.2516 11
0.4008 25
0.45 116
0.45 153
0.3508 41
0.35 49
0.35 33
0.45 75
0.1992 3
0.4 24
0.4 105
0.45 156
0.3492 20
0.3008 17
0.5 264
0.4492 126
0.3992 42
0.45 105
0.4 14
0.45 77
0.45 103
0.35 5
0.45 72
0.45 187
0.4492 69
0.45 136
0.45 62
0.2492 2
0.2492 16
0.0508 1
0.45 115
0.2008 11
0.1492 1
0.05 3
0.0492 3
0.1008 4
0.0992 5
0.0992 6
0.15 8
0.05 2
0.1008 7
0.15 16
0.05 5
0.0492 4
0.0984 7
0.1516 7
0.1492 8
0.0492 2
0.1 1
0.1008 5
0.0508 3
0.0992 3
0.0484 3
0.0508 2
0.1 2
0.1008 8
0.0516 3
0.0984 4
0.2 14
0.1016 4
0.25 24
0.0984 8
0.25 15
0.15 11
0.1492 6
0.1508 10
0.25 26
0.2 18
0.1992 19
0.25 32
0.0484 1
0.2508 16
0.1508 9
0.0508 4
0.15 17
0.25 7
0.1508 5
0.2 15
0.1508 8
0.0484 2
0.2 17
0.0992 2
0.05 6
0.2 23
0.1484 8
0.0484 4
0.1984 11
0.0476 3
0.1016 7
0.2492 20
0.0516 4
0.0516 1
0.1984 9
0.25 29
0.0476 1
0.25 22
0.0516 2
0.1992 5
0.25 34
0.0976 4
0.2008 10
0.2 20
0.1008 2
0.1024 4
0.0476 2
0.0524 2
0.1484 10
0.0468 2
0.0524 3
0.2492 23
0.1992 8
0.2492 18
0.2 25
0.0984 5
0.0532 2
0.1992 25
0.1992 15
0.2016 11
0.0492 6
0.1508 2
0.0468 3
0.2 7
0.0524 1
0.0468 1
0.1508 6
0.0508 6
0.1516 8
0.0976 7
0.0984 2
0.046 2
0.2508 23
0.1476 7
0.1016 5
0.1484 2
0.1008 6
0.1992 7
0.0532 1
0.046 1
0.15 4
0.054 2
0.1016 8
0.1024 7
0.1492 3
0.0968 4
0.15 19
0.15 12
0.45 92
0.3992 40
0.4492 156
0.45 97
0.45 137
0.0452 2
0.45 203
0.2492 5
0.4492 187
0.4 29
0.4508 110
0.45 142
0.4 9
0.4 72
0.45 107
0.45 96
0.4492 73
0.3992 32
0.4008 23
0.45 114
0.45 56
0.4 43
0.45 48
0.3992 14
0.45 111
0.4508 76
0.4492 102
0.4492 66
0.4492 62
0.3492 19
0.1476 8
0.2008 3
0.0532 3
0.1484 5
0.1984 10
0.1516 5
0.2508 20
0.054 1
0.1484 6
0.2492 15
0.1476 5
0.2008 5
0.1508 1
0.2476 11
0.2492 6
0.3492 33
0.1016 2
0.1524 7
0.1524 5
0.1992 4
0.2508 15
0.2008 4
0.2008 8
0.2 12
0.25 10
0.1492 4
0.1984 3
0.2484 16
0.1984 5
0.2 1
0.3508 19
0.35 12
0.3492 8
0.1032 4
0.3492 13
0.3508 13
0.3492 12
0.35 7
0.3 8
0.3992 19
0.3484 13
0.2992 9
0.3984 25
0.1984 4
0.3492 18
0.4 17
0.4 30
0.3 3
0.2508 6
0.45 29
0.3 1
0.4016 25
0.3492 9
0.35 10
0.3508 8
0.35 15
0.4 28
0.3508 12
0.45 134
0.0548 2
0.45 78
0.45 133
0.1984 8
0.3984 23
0.3492 50
0.45 34
0.45 181
0.4492 75
0.45 70
0.45 20
0.45 104
0.45 143
0.4 63
0.5 211
0.45 82
0.4484 76
0.45 13
0.4492 101
0.45 47
0.4 65
0.4492 13
0.4492 115
0.4492 56
0.45 148
0.5 238
0.45 23
0.5 205
0.4 78
0.3992 65
0.4 46
0.4508 80
0.45 130
0.45 95
0.4492 103
0.45 117
0.2 6
0.45 174
0.35 4
0.4492 70
0.45 14
0.4508 115
0.2992 19
0.2492 7
0.0452 1
0.35 56
0.45 99
0.3508 18
0.45 140
0.45 179
0.3 18
0.1492 12
0.4 52
0.4492 82
0.4508 75
0.45 87
0.45 106
0.4 74
0.35 63
0.35 35
0.4 41
0.3992 9
0.3 16
0.2016 9
0.3 36
0.0968 7
0.3992 64
0.3 20
0.45 65
0.25 1
0.35 52
0.4 79
0.3508 9
0.4 80
0.4508 73
0.3 5
0.45 129
0.3508 29
0.2992 18
0.2992 36
0.4 37
0.0444 2
0.4 69
0.3008 24
0.2992 29
0.3992 41
0.2992 12
0.4008 64
0.4 92
0.4 86
0.1032 7
0.4 66
0.45 86
0.2 16
0.3508 20
0.35 72
0.3492 56
0.3484 41
0.35 14
0.1992 16
0.4 62
0.35 59
0.2492 29
0.3992 69
0.0556 2
0.3492 39
0.35 69
0.3 25
0.1992 13
0.45 155
0.1508 4
0.2008 7
0.3008 29
0.2016 10
0.45 59
0.4484 80
0.3492 63
0.3 45
0.1976 9
0.25 3
0.25 13
0.2992 25
0.0548 1
0.3992 72
0.3 26
0.2492 30
0.2992 49
0.1008 3
0.3492 69
0.35 42
0.2516 16
0.2984 29
0.2984 24
0.1992 12
0.0444 1
0.2492 12
0.4 117
0.1508 12
0.2508 29
0.1992 20
0.35 53
0.3 33
0.3008 12
0.2992 33
0.2492 24
0.3008 25
0.35 60
0.3 51
0.2016 4
0.2508 18
0.0984 3
0.3508 56
0.2016 8
0.1516 6
0.3016 29
0.2976 29
0.2508 17
0.1492 11
0.25 25
0.3992 46
0.2508 24
0.3008 33
0.1992 21
0.2484 24
0.3992 80
0.2008 13
0.0484 6
0.1508 11
0.1484 11
0.0976 2
0.1484 4
0.3008 49
0.35 48
0.25 27
0.4008 80
0.2492 27
0.0556 1
0.3 34
0.2492 31
0.1984 13
0.2 24
0.1516 4
0.1976 4
0.2484 20
0.2492 13
0.2008 19
0.1976 10
0.1524 8
0.2508 7
0.096 4
0.0436 1
0.2492 26
0.1992 14
0.1468 8
0.3 22
0.0976 5
0.1476 4
0.3 47
0.0564 1
0.1492 17
0.15 14
0.1984 19
0.1024 5
0.0476 4
0.1976 11
0.3008 36
0.2492 22
0.3492 37
0.2484 15
0.2508 26
0.1468 7
0.25 28
0.1516 10
0.2492 9
0.35 68
0.2492 25
0.2508 13
0.104 4
0.35 65
0.2024 4
0.2484 17
0.2992 34
0.1532 8
0.0968 5
0.1524 4
0.2008 12
0.1532 7
0.35 36
0.2516 17
0.146 7
0.2484 18
0.046 3
0.3 44
0.1984 12
0.0428 1
0.2984 36
0.4 26
0.4 51
0.4492 87
0.4 22
0.3484 18
0.3492 7
0.4492 79
0.3484 8
0.3492 14
0.3484 19
0.35 16
0.4 77
0.4 27
0.35 34
0.4008 9
0.3992 60
0.4008 46
0.4492 99
0.4008 47
0.45 42
0.4484 83
0.45 49
0.45 71
0.4492 58
0.45 124
0.45 88
0.4484 73
0.4492 47
0.45 53
0.4492 92
0.4508 69
0.4508 66
0.3992 37
0.4508 56
0.45 84
0.45 85
0.45 40
0.4492 48
0.4516 73
0.4492 77
0.4508 77
0.5 92
0.4492 65
0.45 100
0.45 109
0.4508 70
0.4516 76
0.4508 92
0.4508 62
0.4492 71
0.4484 92
0.45 67
0.4508 82
0.45 90
0.4516 80
0.4492 85
0.4484 82
0.4492 78
0.4492 84
0.4508 85
0.4484 85
0.4492 91
0.4508 78
0.4492 86
0.4508 79
0.5 96
0.4492 95
0.4516 82
0.4492 88
0.4508 99
0.4508 86
0.4492 93
0.5 107
0.4492 90
0.5 91
0.4516 92
0.4508 93
0.4516 83
0.4476 92
0.4492 96
0.4484 99
0.4508 95
0.4492 105
0.4476 83
0.4492 94
0.4484 93
0.4508 96
0.4508 94
0.4992 92
0.4508 90
0.4508 88
0.4508 91
0.4492 97
0.4516 93
0.4484 94
0.4508 102
0.4508 103
0.4508 97
0.4484 97
0.4484 95
0.5 100
0.4484 102
0.4508 101
0.4484 103
0.4484 96
0.4516 99
0.4484 101
0.4516 103
0.4516 101
0.4476 101
0.4492 100
0.4476 99
0.5 99
0.5 101
0.4508 100
0.4516 94
0.4476 94
0.4516 102
0.4476 102
0.4524 101
0.4484 100
0.45 98
0.4516 100
0.4476 100
0.4524 100
0.4468 101
0.4468 100
0.4532 100
0.4492 98
0.4508 98
0.446 100
0.454 100
0.4452 100
0.4532 101
0.4548 100
0.4444 100
0.4556 100
0.4992 100
0.5008 100
0.446 101
};
\end{axis}
\begin{axis}[
color=red,  label style={font=\Large},
separate axis lines,
  axis y line*=left,
  axis x line=none,
  ylabel=number of elections,
  ymin=-9, ymax=260,
  xmin=0, xmax=0.52,ylabel shift = -5pt,every axis 
plot/.append style={thick}
]
\addplot[red,mark=*] coordinates{
	(0.05, 54) (0.1, 47) (0.15, 68)
(0.2, 71) (0.25, 76) (0.3, 59) (0.35, 82)
(0.4, 89) (0.45, 230) (0.5, 14)
};
\end{axis}

\end{tikzpicture}

%% file: l_image.tex
\definecolor{color0}{rgb}{1,0.85,0.2}
\definecolor{color1}{rgb}{0.95,0.7,0}
\definecolor{color2}{rgb}{1,0.5,0}
\definecolor{color3}{rgb}{1,0.25,0}
\definecolor{color4}{rgb}{1,0,0}
\definecolor{color5}{rgb}{0.85,0,0.75}
\definecolor{color6}{rgb}{0.625,0,0.57}
\definecolor{color7}{rgb}{0.4,0,0.4}
\definecolor{color8}{rgb}{0,0,0}
\definecolor{color9}{rgb}{0.41,0.41,0.41}

\begin{tikzpicture}
\node[anchor=north east, inner sep=0pt] (russell) at (-1.25,0)
{};
\node[anchor=north east, inner sep=0pt] (russell) at (16,0)
{};
\node[anchor=north east, inner sep=0pt] (russell) at (0,0)
{\resizebox{1.25cm}{!}{\input{id110.tex}}};
\node[anchor=north east,inner sep=0pt] (russell) at (3.6,0)
{\resizebox{1.25cm}{!}{\input{id255.tex}}};
\node[anchor=north east,inner sep=0pt] (russell) at (0,-3.3)
{\resizebox{1.25cm}{!}{\input{id42.tex}}};
\node[anchor=north east,inner sep=0pt] (russell) at (3.6,-3.3)
{\resizebox{1.25cm}{!}{\input{id149.tex}}};
\node[anchor=north east,inner sep=0pt] (russell) at (14.3,0)
{\resizebox{1.25cm}{!}{\input{id537.tex}}};
\node[anchor=north east,inner sep=0pt] (russell) at (14.3,-3.3)
{\resizebox{1.25cm}{!}{\input{id491.tex}}};
\node[anchor=north east,inner sep=0pt] (whitehead) at (12.3,0.3)
{\includegraphics[width=7cm]{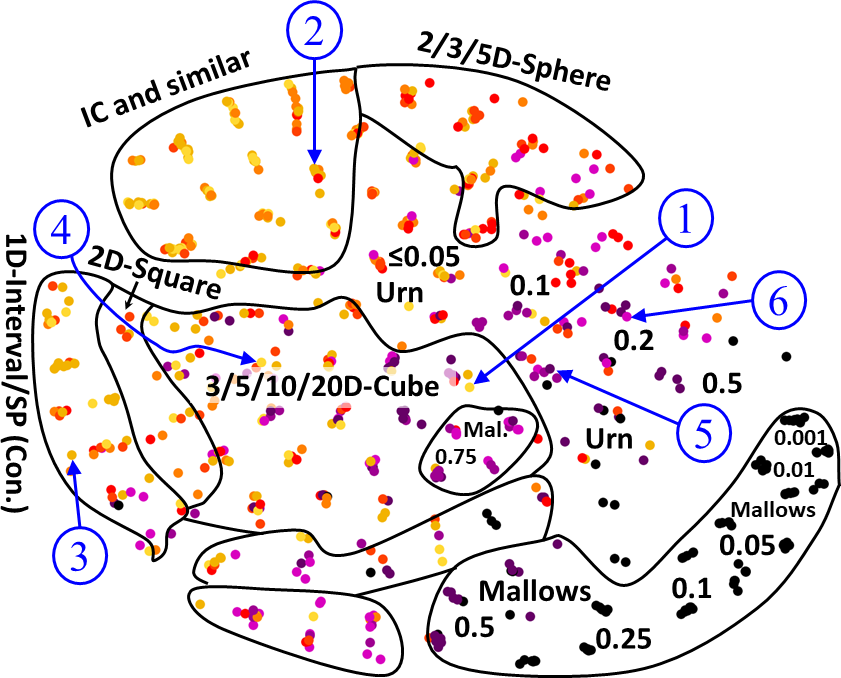}};
\setlength\tabcolsep{-3pt}
\node[thick,rounded corners=2pt,below left=2mm] at 
(12.65,-5.25) 
{%
	\resizebox{7cm}{!}{\begin{tabular}{cccccccccc}
		\raisebox{2pt}{\tikz{\draw[line width=3mm,color0] (0,0) -- 
				(10mm,0);}}& \raisebox{2pt}{\tikz{\draw[line width=3mm,color1] 
				(0,0) -- 
				(10mm,0);}}& \raisebox{2pt}{\tikz{\draw[line width=3mm,color2] 
				(0,0) -- 
				(10mm,0);}}&  \raisebox{2pt}{\tikz{\draw[line width=3mm,color3] 
				(0,0) -- 
				(10mm,0);}}& \raisebox{2pt}{\tikz{\draw[line width=3mm,color4] 
				(0,0) -- 
				(10mm,0);}} & \raisebox{2pt}{\tikz{\draw[line width=3mm,color5] 
				(0,0) -- 
				(10mm,0);}}& \raisebox{2pt}{\tikz{\draw[line width=3mm,color6] 
				(0,0) -- 
				(10mm,0);}}& \raisebox{2pt}{\tikz{\draw[line width=3mm,color7] 
				(0,0) -- 
				(10mm,0);}}& \raisebox{2pt}{\tikz{\draw[line width=3mm,color8] 
				(0,0) -- 
				(10mm,0);}} \\
		$0.05$ & $0.1$ & $0.15$ & $0.2$ & $0.25$ & $0.3$ & $0.35$ & $0.4$ & 
		$0.45$
		\end{tabular}}};
\node[circle,inner sep=1.5pt,draw=blue,line width=0.7pt] at (1.2,0.2)  
{\textcolor{blue}{1}};
\node[circle,inner sep=1.5pt,draw=blue,line width=0.7pt] at (2.4,0.2)  
{\textcolor{blue}{2}};
\node[circle,inner sep=1.5pt,draw=blue,line width=0.7pt] at (1.2,-3.1)  
{\textcolor{blue}{3}};
\node[circle,inner sep=1.5pt,draw=blue,line width=0.7pt] at (2.4,-3.1)  
{\textcolor{blue}{4}};
\node[circle,inner sep=1.5pt,draw=blue,line width=0.7pt] at (13,0.2)  
{\textcolor{blue}{5}};
\node[circle,inner sep=1.5pt,draw=blue,line width=0.7pt] at (13,-3.1)  
{\textcolor{blue}{6}};
\draw (4.6,-0.1) -- (4.6, 0.4);
\node[text width=6cm] at (6,0.45) 
{\scriptsize{Plurality-score/Borda-score/rank}};
\end{tikzpicture}

%% file: l_image_borda.tex
\definecolor{color0}{rgb}{1,0.85,0.2}
\definecolor{color1}{rgb}{0.95,0.7,0}
\definecolor{color2}{rgb}{1,0.5,0}
\definecolor{color3}{rgb}{1,0.25,0}
\definecolor{color4}{rgb}{1,0,0}
\definecolor{color5}{rgb}{0.85,0,0.75}
\definecolor{color6}{rgb}{0.625,0,0.57}
\definecolor{color7}{rgb}{0.4,0,0.4}
\definecolor{color8}{rgb}{0,0,0}
\definecolor{color9}{rgb}{0.41,0.41,0.41}

\begin{tikzpicture}
\node[anchor=north east, inner sep=0pt] (russell) at (-1.25,0)
{};
\node[anchor=north east, inner sep=0pt] (russell) at (16,0)
{};
\node[anchor=north east, inner sep=0pt] (russell) at (0,0)
{\resizebox{1.25cm}{!}{\input{id661.tex}}};
\node[anchor=north east,inner sep=0pt] (russell) at (3.6,0)
{\resizebox{1.25cm}{!}{\input{id40.tex}}};
\node[anchor=north east,inner sep=0pt] (russell) at (0,-3.3)
{\resizebox{1.25cm}{!}{\input{id504.tex}}};
\node[anchor=north east,inner sep=0pt] (russell) at (3.6,-3.3)
{\resizebox{1.25cm}{!}{\input{id307.tex}}};
\node[anchor=north east,inner sep=0pt] (russell) at (14.3,0)
{\resizebox{1.25cm}{!}{\input{id20.tex}}};
\node[anchor=north east,inner sep=0pt] (russell) at (14.3,-3.3)
{\resizebox{1.25cm}{!}{\input{id459.tex}}};
\node[anchor=north east,inner sep=0pt] (whitehead) at (12.3,0.3)
{\includegraphics[width=7cm]{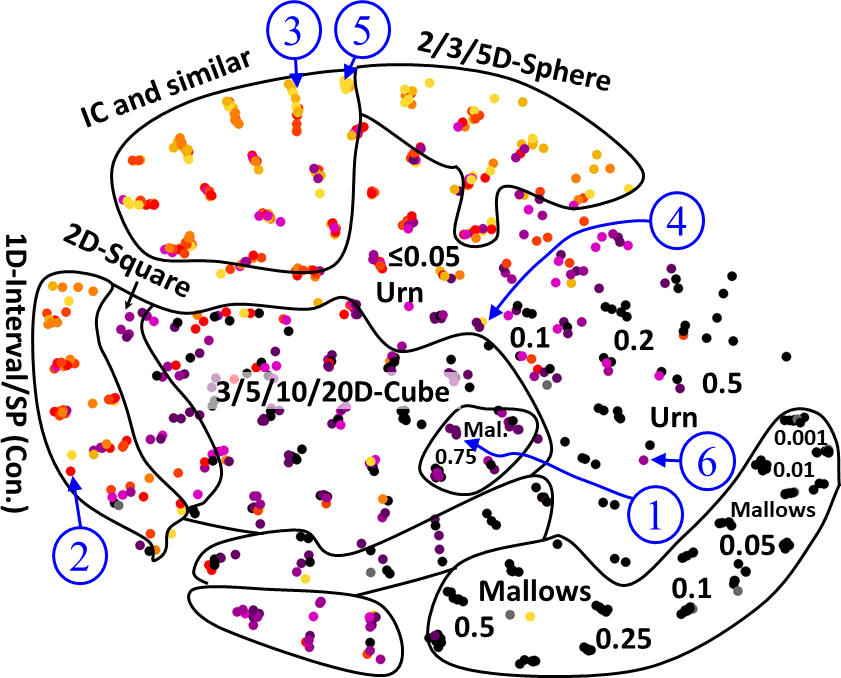}};
\setlength\tabcolsep{-3pt}
\node[thick,rounded corners=2pt,below left=2mm] at 
(12.65,-5.25) 
{%
	\resizebox{7cm}{!}{\begin{tabular}{ccccccccccc}
		\raisebox{2pt}{\tikz{\draw[line width=2.7mm,color0] (0,0) -- 
				(10mm,0);}}& \raisebox{2pt}{\tikz{\draw[line 
				width=2.7mm,color1] (0,0) -- 
				(10mm,0);}}& \raisebox{2pt}{\tikz{\draw[line 
				width=2.7mm,color2] (0,0) -- 
				(10mm,0);}}&  \raisebox{2pt}{\tikz{\draw[line 
				width=2.7mm,color3] (0,0) -- 
				(10mm,0);}}& \raisebox{2pt}{\tikz{\draw[line 
				width=2.7mm,color4] (0,0) -- 
				(10mm,0);}} & \raisebox{2pt}{\tikz{\draw[line 
				width=2.7mm,color5] (0,0) -- 
				(10mm,0);}}& \raisebox{2pt}{\tikz{\draw[line 
				width=2.7mm,color6] (0,0) -- 
				(10mm,0);}}& \raisebox{2pt}{\tikz{\draw[line 
				width=2.7mm,color7] (0,0) -- 
				(10mm,0);}}& \raisebox{2pt}{\tikz{\draw[line 
				width=2.7mm,color8] (0,0) -- 
				(10mm,0);}}& \raisebox{2pt}{\tikz{\draw[line 
				width=2.7mm,color9] (0,0) -- 
				(10mm,0);}} \\
		$0.05$ & $0.1$ & $0.15$ & $0.2$ & $0.25$ & $0.3$ & $0.35$ & $0.4$ & 
		$0.45$ & 
		$0.5$
		\end{tabular}}};
\node[circle,inner sep=1.5pt,draw=blue,line width=0.7pt] at (1.2,0.2)  
{\textcolor{blue}{1}};
\node[circle,inner sep=1.5pt,draw=blue,line width=0.7pt] at (2.4,0.2)  
{\textcolor{blue}{2}};
\node[circle,inner sep=1.5pt,draw=blue,line width=0.7pt] at (1.2,-3.1)  
{\textcolor{blue}{4}};
\node[circle,inner sep=1.5pt,draw=blue,line width=0.7pt] at (2.4,-3.1)  
{\textcolor{blue}{5}};
\node[circle,inner sep=1.5pt,draw=blue,line width=0.7pt] at (13,0.2)  
{\textcolor{blue}{3}};
\node[circle,inner sep=1.5pt,draw=blue,line width=0.7pt] at (13,-3.1)  
{\textcolor{blue}{6}};
\draw (4.6,-0.1) -- (4.6, 0.4);
\node[text width=6cm] at (6.8,0.45) 
{\scriptsize{Borda-score/rank}};
\end{tikzpicture}

%% file: 50w_full_plur.tex
\begin{tikzpicture}
\begin{axis}[
    enlargelimits=false,xtick={0.05,0.1,0.15,0.2,0.25,0.3,0.35,0.4,0.45,0.5},
xticklabels={0.05,0.1,0.15,0.2,0.25,0.3,0.35,0.4,0.45,0.5},ytick={0,50,100,150,
200} , yticklabels={,50,100, 150,200},
	xlabel={50\%-winner (relative)},
	ylabel={number of elections},legend pos=north east,legend columns=2, 
legend style={draw=none,/tikz/column 2/.style={column sep=5pt,}, legend style 
= {font=\large,fill=none},  every axis plot/.append style={thick}}
]
\addplot [mark=, red] coordinates {
	(0.049999999999999996, 34) (0.09999999999999999, 38) (0.15, 61) 
(0.19999999999999998, 62) (0.24999999999999997, 64) (0.3, 70) (0.35, 69) 
(0.39999999999999997, 95) (0.44999999999999996, 149) (0.49999999999999994, 118)
};

\addplot [mark=, blue]  coordinates{
	(0.049999999999999996, 67) (0.09999999999999999, 126) (0.15, 92) 
(0.19999999999999998, 73) (0.24999999999999997, 67) (0.3, 67) (0.35, 51) 
(0.39999999999999997, 69) (0.44999999999999996, 123) (0.49999999999999994, 0)
};

\addplot [mark=, black] coordinates{
	(0.049999999999999996, 205) (0.09999999999999999, 157) (0.15, 69) 
(0.19999999999999998, 54) (0.24999999999999997, 41) (0.3, 37) (0.35, 82) 
(0.39999999999999997, 67) (0.44999999999999996, 0) (0.49999999999999994, 0)
};
\legend{n-100 m-5,n-100 m-10,n-100 m-20}
\end{axis}
\end{tikzpicture}

%% file: 50w_full_bord.tex
\begin{tikzpicture}
\begin{axis}[
    enlargelimits=false,xtick={0.05,0.1,0.15,0.2,0.25,0.3,0.35,0.4,0.45,0.5},
xticklabels={0.05,0.1,0.15,0.2,0.25,0.3,0.35,0.4,0.45,0.5},ytick={0,50,100,150,
200} , yticklabels={,50,100, 150,200} ,
	xlabel={50\%-winner (relative)},
	ylabel={number of elections},legend pos=north west,legend columns=2, 
legend style={draw=none,/tikz/column 2/.style={column sep=5pt,}, legend style 
= {font=\large,fill=none},  every axis plot/.append style={thick}}
]
\addplot [mark=, red] coordinates {
	(0.049999999999999996, 15) (0.09999999999999999, 29) (0.15, 33) 
(0.19999999999999998, 42) (0.24999999999999997, 56) (0.3, 79) (0.35, 74) 
(0.39999999999999997, 119) (0.44999999999999996, 171) (0.49999999999999994, 172)
};

\addplot [mark=, blue]  coordinates{
(0.049999999999999996, 54) (0.09999999999999999, 47) (0.15, 68) 
(0.19999999999999998, 71) (0.24999999999999997, 76) (0.3, 59) (0.35, 82) 
(0.39999999999999997, 89) (0.44999999999999996, 230) (0.49999999999999994, 14)
};

\addplot [mark=, black] coordinates{
	(0.049999999999999996, 89) (0.09999999999999999, 96) (0.15, 61) 
(0.19999999999999998, 72) (0.24999999999999997, 67) (0.3, 56) (0.35, 89) 
(0.39999999999999997, 191) (0.44999999999999996, 68) (0.49999999999999994, 0)
};
\legend{n-100 m-5,n-100 m-10,n-100 m-20}
\end{axis}
\end{tikzpicture}

%% file: l_image_append.tex
\definecolor{color0}{rgb}{1,0.85,0.2}
\definecolor{color1}{rgb}{0.95,0.7,0}
\definecolor{color2}{rgb}{1,0.5,0}
\definecolor{color3}{rgb}{1,0.25,0}
\definecolor{color4}{rgb}{1,0,0}
\definecolor{color5}{rgb}{0.85,0,0.75}
\definecolor{color6}{rgb}{0.625,0,0.57}
\definecolor{color7}{rgb}{0.4,0,0.4}
\definecolor{color8}{rgb}{0,0,0}
\definecolor{color9}{rgb}{0.41,0.41,0.41}

\begin{tikzpicture}
\node[anchor=north east,inner sep=0pt] (whitehead) at (0,14)
{\includegraphics[width=8cm]{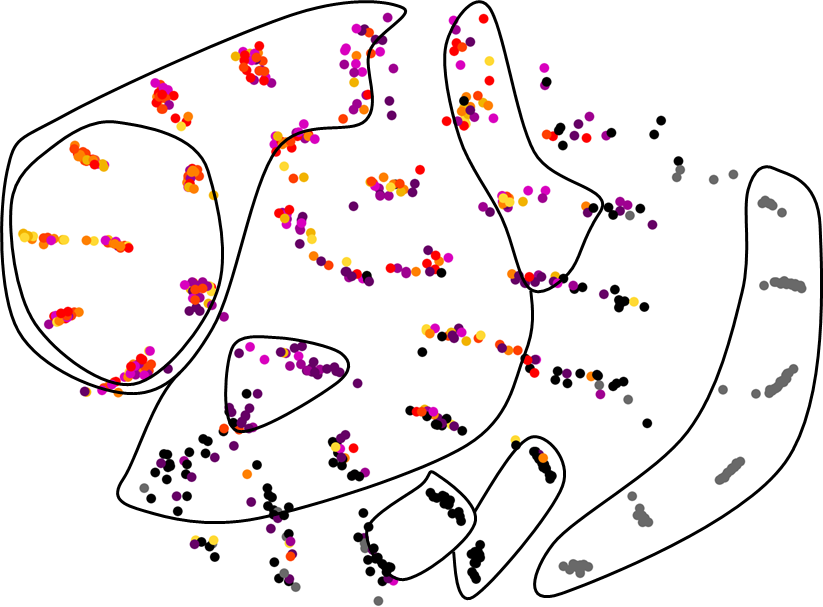}};
\node[anchor=north east,inner sep=0pt] (whitehead) at (0,7)
{\includegraphics[width=8cm]{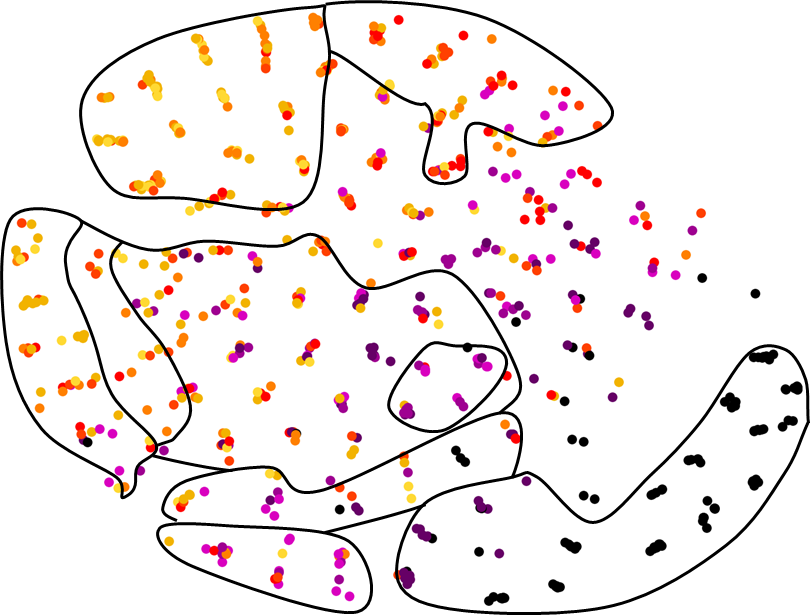}}; 
\node[anchor=north east,inner sep=0pt] (whitehead) at (0,0)
{\includegraphics[width=8cm]{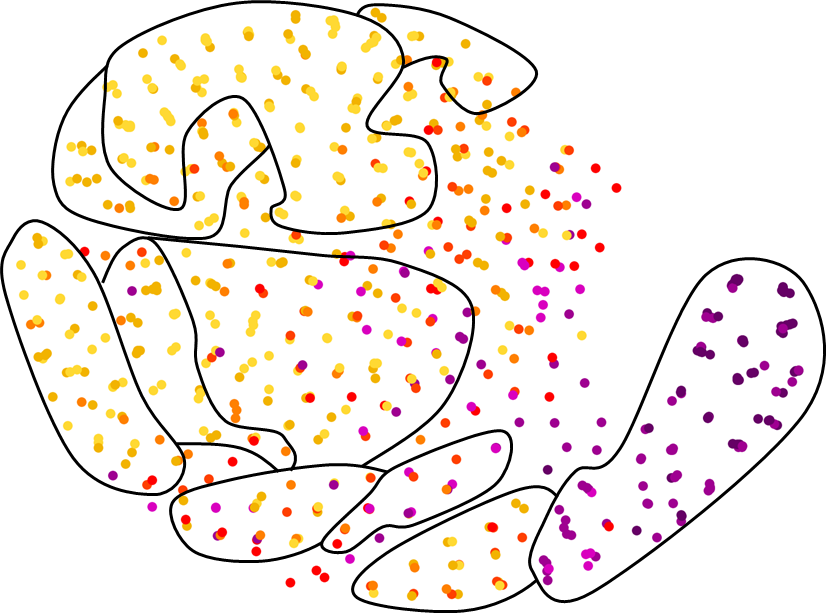}}; 
\node[anchor=north east,inner sep=0pt] (whitehead) at (8.5,14)
{\includegraphics[width=8cm]{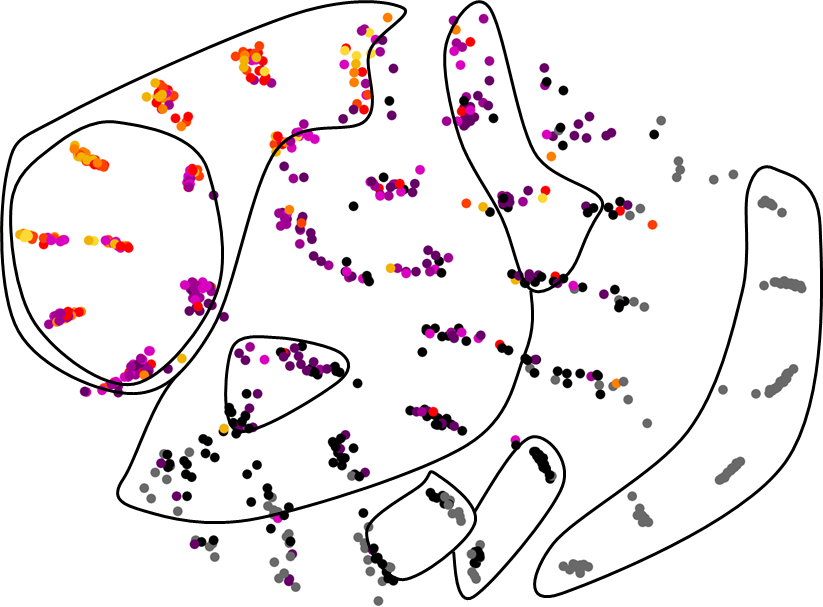}};
\node[anchor=north east,inner sep=0pt] (whitehead) at (8.5,7)
{\includegraphics[width=8cm]{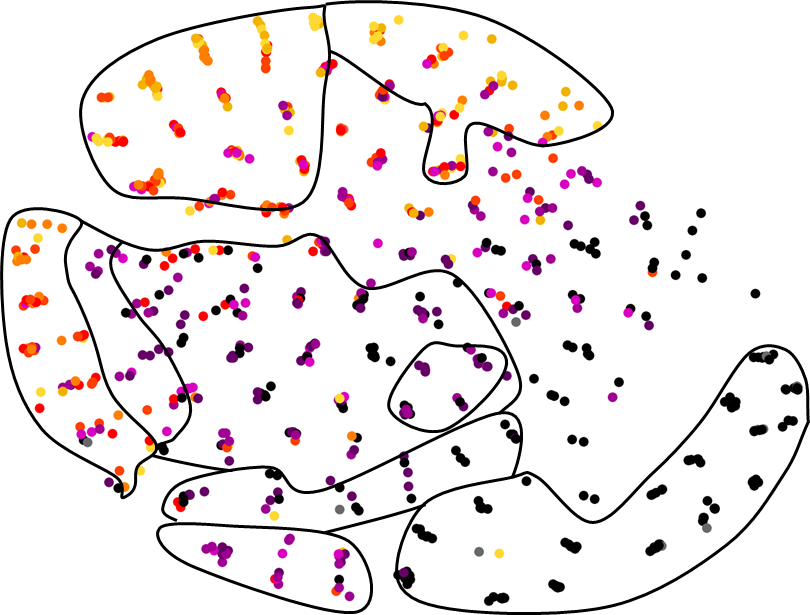}}; 
\node[anchor=north east,inner sep=0pt] (whitehead) at (8.5,0)
{\includegraphics[width=8cm]{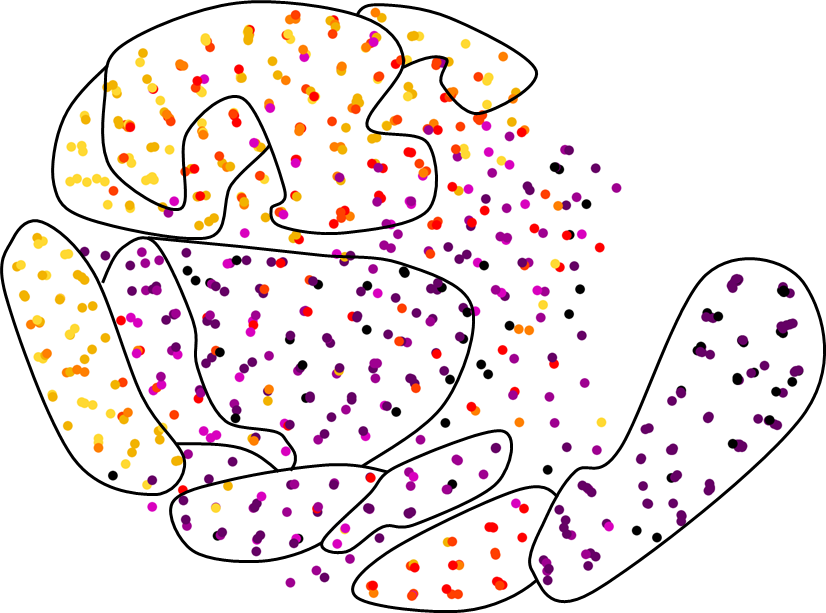}}; 
\node[anchor=north east,inner sep=0pt] (whitehead) at (17,14)
{\includegraphics[width=8cm]{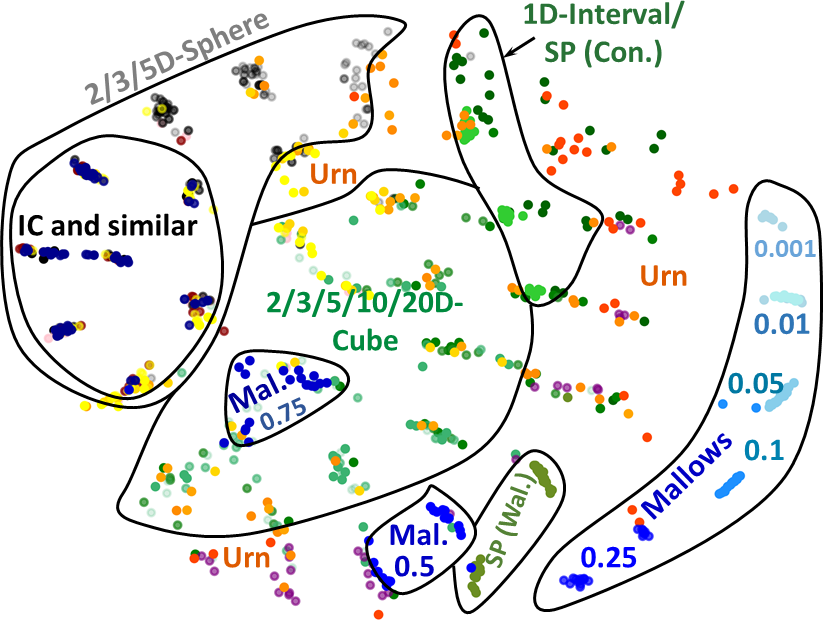}};
\node[anchor=north east,inner sep=0pt] (whitehead) at (17,7)
{\includegraphics[width=8cm]{map10mask.png}}; 
\node[anchor=north east,inner sep=0pt] (whitehead) at (17,0)
{\includegraphics[width=8cm]{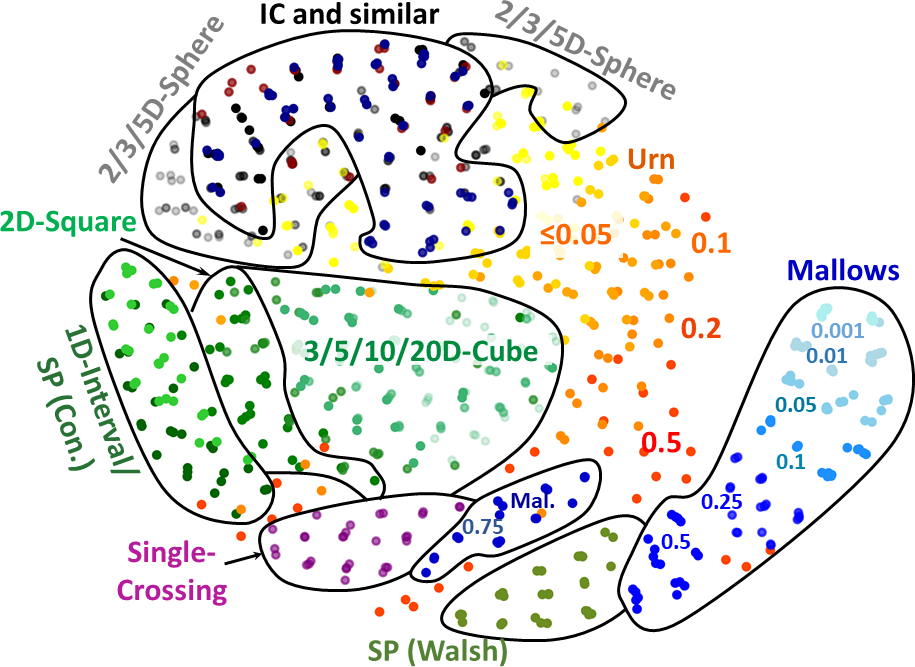}}; 
\setlength\tabcolsep{-3pt}
\node[thick,rounded corners=2pt,below left=2mm] at 
(8,-7) 
{%
	\resizebox{14cm}{!}{\begin{tabular}{ccccccccccc}
		\raisebox{2pt}{\tikz{\draw[line width=2.7mm,color0] (0,0) -- 
				(10mm,0);}}& \raisebox{2pt}{\tikz{\draw[line 
				width=2.7mm,color1] (0,0) -- 
				(10mm,0);}}& \raisebox{2pt}{\tikz{\draw[line 
				width=2.7mm,color2] (0,0) -- 
				(10mm,0);}}&  \raisebox{2pt}{\tikz{\draw[line 
				width=2.7mm,color3] (0,0) -- 
				(10mm,0);}}& \raisebox{2pt}{\tikz{\draw[line 
				width=2.7mm,color4] (0,0) -- 
				(10mm,0);}} & \raisebox{2pt}{\tikz{\draw[line 
				width=2.7mm,color5] (0,0) -- 
				(10mm,0);}}& \raisebox{2pt}{\tikz{\draw[line 
				width=2.7mm,color6] (0,0) -- 
				(10mm,0);}}& \raisebox{2pt}{\tikz{\draw[line 
				width=2.7mm,color7] (0,0) -- 
				(10mm,0);}}& \raisebox{2pt}{\tikz{\draw[line 
				width=2.7mm,color8] (0,0) -- 
				(10mm,0);}}& \raisebox{2pt}{\tikz{\draw[line 
				width=2.7mm,color9] (0,0) -- 
				(10mm,0);}} \\
		$0.05$ & $0.1$ & $0.15$ & $0.2$ & $0.25$ & $0.3$ & $0.35$ & $0.4$ & 
		$0.45$ & 
		$0.5$
		\end{tabular}}};
\node[text width=6cm] at (6.5,15) 
{\large{Borda}};
\node[text width=6cm] at (-2,15) 
{\large{Plurality}};
\node[text width=6cm] at (14.5,15) 
{\large{Statistical cultures}};

\node[text width=6cm] at (-5,13.7) 
{\large{5 candidates}};
\node[text width=6cm] at (-5,7) 
{\large{10 candidates}};
\node[text width=6cm] at (-5,0) 
{\large{20 candidates}};
\draw[line width=0.4mm] (8.7,-8.3) -- (8.7, 16);
\end{tikzpicture}